\documentclass[aps,pra,10pt,twocolumn,showpacs,superscriptaddress,preprintnumbers]{revtex4-1}

\usepackage[utf8]{inputenc}  
\usepackage[T1]{fontenc}
\usepackage[english]{babel}  

\usepackage{graphicx}
\usepackage{bm}

\usepackage{amsmath,amssymb,amscd,amsthm} 

\usepackage{color}
\usepackage{verbatim} 
\usepackage{xspace}
\usepackage{rotating}

\usepackage[bookmarks=false,pdfstartview={FitH}]{hyperref}

\theoremstyle{plain}
\newtheorem{theorem}{Theorem}
\newtheorem{corollary}{Corollary}

\theoremstyle{definition}
\newtheorem{example}{Example}

\renewcommand{\Re}{\operatorname{Re}}
\newcommand{\unity}{\ensuremath{{\rm 1 \negthickspace l}{}}}

\newcommand{\trace}{\operatorname{tr}}
\newcommand{\diag}{\operatorname{diag}}

\newcommand{\spec}{\sigma} 
\newcommand{\wrt}[1]{\:\mathrm{d}#1\:} 

\newcommand{\bra}[1]{\ensuremath{\langle #1 |}{}}
\newcommand{\ket}[1]{\ensuremath{| #1 \rangle}{}}

\newcommand{\ketbra}[2]{\ensuremath{| #1 \rangle \langle #2 |}{}}

\newcommand{\expt}[1]{\ensuremath{\langle #1 \rangle}{}}

\newcommand{\SWAP}{{\sc swap}}
\newcommand{\iSWAP}{i{\sc swap}}

\newcommand{\Partial}[2]{\ensuremath \frac{\partial{#1}}{\partial{#2}}{}}

\newcommand{\sys}[1]{\ensuremath{_{#1}}} 

\newcommand{\tauf}{\tau}

\newcommand{\su}{\mathfrak{su}}

\newcommand{\MSC}{\ensuremath{\sf{MSC}}\xspace}
\newcommand{\MMC}{\ensuremath{\sf{MMC}}\xspace}
\newcommand{\KSC}{\ensuremath{\sf{KSC}}\xspace}
\newcommand{\KMC}{\ensuremath{\sf{KMC}}\xspace}
\newcommand{\DSC}{\ensuremath{\sf{DSC}}\xspace}

\newcommand{\reach}{\mathfrak{reach}}
\newcommand{\pos}{\mathfrak{pos}}

\newcommand{\grape}{{\sc grape}\xspace}
\newcommand{\dynamo}{{\sc dynamo}\xspace}

\newcommand{\be}{\begin{equation}}
\newcommand{\ee}{\end{equation}}

\newcommand{\expfactorbitflip}{\varepsilon}
\newcommand{\expfactoramp}{\varepsilon}

\newcommand{\pla}{n(\omega_0)}
\newcommand{\boltz}{e^{-\beta \hbar \omega_0}}
\newcommand{\bboltz}{b}                
\newcommand{\rhoalg}{\rho_{\text{alg}}}

\newcommand{\Hlamb}{H_{\text{LS}}}     
\newcommand{\Hrf}{H_{\text{rot}}}      
\newcommand{\carrier}{\tilde{\omega}}  
\newcommand{\omegacut}{\omega_{\text{cut}}}  

\newcommand{\speccorr}{C}  

\newcommand{\I}{\unity} 

\begin{document}


\title{Arbitrary $n$-Qubit State Transfer Implemented by\\ Coherent Control and Simplest Switchable Local Noise}

\author{Ville Bergholm}
\email{ville.bergholm@iki.fi}
\affiliation{Dept.~Chemistry, Technical University of Munich (TUM), D-85747 Garching, Germany}
\author{Frank K.~Wilhelm}\email{fwm@lusi.uni-sb.de}
\affiliation{Institute for Theoretical Physics, University of Saarland, 66123 Saarbr{\"u}cken, Germany}
\author{Thomas Schulte-Herbr{\"u}ggen}\email{tosh@ch.tum.de}
\affiliation{Dept.~Chemistry, Technical University of Munich (TUM), D-85747 Garching, Germany}

\date{\today}


\begin{abstract}
We study reachable sets of open $n$-qubit quantum systems, whose coherent parts
are under full unitary control, by adding as a further degree of incoherent control
switchable Markovian noise on a single qubit.
In particular, adding bang-bang control of amplitude
damping noise (non-unital) allows the dynamic system to act transitively on the entire 
set of density operators. Thus one can transform any initial quantum state into 
any desired target state. Adding switchable bit-flip noise (unital) instead 
suffices to get all states majorised by the initial state. Our 
open-loop optimal control package \dynamo is extended by incoherent control to
exploit these unprecedented reachable sets in experiments. 
We propose implementation by a GMon, a superconducting device
with fast tunable coupling to an open transmission line, and illustrate
how open-loop control with noise switching achieves
all state transfers without
measurement-based closed-loop feedback
and resettable ancilla.
\end{abstract}

\maketitle



Recently, dissipation was exploited for quantum state
engineering~\cite{VWC09,KMP11} so that evolution under constant noise leads
to long-lived entangled fixed-point states. Earlier, Lloyd and Viola~\cite{VioLloyd01} showed 
that closed-loop feedback from one {\em resettable ancilla qubit} suffices to simulate any 
quantum dynamics of open systems. 
Both concepts were used to combine coherent  dynamics with optical pumping
on an ancilla qubit for dissipative preparation of entangled states~\cite{BZB11} or quantum maps~\cite{SMB12}. 
Clearly full control over the Kraus operators~\cite{Rabitz07b} 
or the environment~\cite{Pechen11} allows for interconverting arbitrary
quantum states.

Manipulating quantum systems with high precision is paramount to exploring their
properties for pioneering experiments, e.g., in view of
new technologies~\cite{DowMil03, control-roadmap2015}.
Superconducting qubits count among the most promising designs for 
scalable
quantum simulation and quantum information processing.
First adjustable couplers were introduced in flux qubits~\cite{Hime_2006, Harris_2007}.
Recently, {\em fast tunable couplers} were implemented for transmon qubits, e.g.,
in the GMon design~\cite{Mart09,Mart13,Mart14}.
Thus the goal to extend the current toolbox of optimal control~\cite{WisMil09, PRA11,control-roadmap2015}
by
dissipative controls has come within reach.

In this letter, first we prove that it suffices to include as a new control parameter 
a single bang-bang switchable Markovian noise amplitude on one qubit (no ancilla)
into an otherwise noiseless and coherently controllable network to increase the power
of the dynamic system so that {\em any target state can be reached from any initial state}.
We then study several state transfer problems using our numerical
optimal control platform {\sc dynamo}~\cite{PRA11} extended by
controlled Markovian noise.
Ultimately we propose an experimental implementation of this control method by a chain of GMons
with a tunable coupling to an open transmission line as are now available. We demonstrate numerically
the initialisation, erasure and preparation steps~\cite{VincCriteria} of quantum computing,
as well as noise-assisted generation of maximally entangled states.

\medskip
{\em Overview and Theory.}
We treat the quantum Markovian master equation~\footnote{\label{f:markovianity}
	We are well aware of different notions of Markovianity,
	see for instance~\cite{Plenio_NMarkov_Review2014, Breuer_NMarkov_Review2016}.
        Here, in line with the seminal 
	work of Wolf, Cirac {\em et al.} \cite{Wolf08a,Wolf08b},
	we adopt the following mathematical definition of Markovianity:
	A quantum map $F(t)$ is time-dependent (resp. time-independent) Markovian, if it is the solution of
	a time-dependent (resp. time-independent) Lindblad master equation
	$$ \dot F(t) = -(iH_u(t) + \Gamma_L(t)) F(t).$$
	The above ensures that (in the connected component) $F(t)$ is by construction infinitesimally 
	(resp. infinitely) divisible into \protect{{\sc cptp}}-maps
	in the terminology of Refs.~\cite{Wolf08a,Wolf08b,DHKS08}
	and hence Markovian. In other words, $F(t)$ can be exponentially constructed as a Lie
	semigroup \cite{DHKS08} and thus has no memory terms.
	This notion is precise and well defined in the sense of not invoking approximations at the level of definition.
	In a second step (extensively discussed in the Supplement), for a {\em physical realisation} 
	we check to which extent a Markov approximation in line with, e.g.,
	Refs.~\cite{BreuPetr02,ABLZ12} does hold on the
	{\em operational level of a concrete experimental setting}. ---
        In contrast note that even in the connected component there are Kraus maps, which are {\em not}
	solutions of a time-dependent (or time-independent) Lindblad master equation~\cite{Wolf08a,DHKS08},
        but of more general recent master equations~\cite{DiosiFerialdi_NMarkovMasterEqn2014,
	Ferialdi_NMarkovMasterEqn2016, Vacc_NMarkovMasterEqn2016,
	Ferialdi_NMarkovSpinBosonJayCumm2017}.
	These Kraus maps then are truely non-Markovian in the sense of going beyond time-dependent
	Markovian maps.
	} 
of an $n$-qubit system
as a bilinear control system~$\Sigma$:
\begin{equation}\label{eqn:master}
\dot\rho(t) = -(i\hat H_u +\Gamma)\rho(t)\quad\text{and}\quad\rho(0)=\rho_0
\end{equation}
with $H_u:= H_0 + \sum_j u_j(t) H_j$ comprising the free-evolution Hamiltonian $H_0$, the control
Hamiltonians $H_j$ switched by piecewise constant control amplitudes $u_j(t)\in\mathbb R$ and $\hat H_u$
as the corresponding commutator superoperator.
Take $\Gamma$ to be of Lindblad form
\begin{equation}\label{eqn:Lindblad}
\Gamma(\rho) := -\sum_\ell \gamma_\ell(t)\big(V_\ell\rho V_\ell^\dagger - 
	\tfrac{1}{2}(V^\dagger_\ell V_\ell \rho + \rho V^\dagger_\ell V_\ell)\big)\;,
\end{equation}
where now $\gamma_\ell(t)\in[0,\gamma_*]$ with $\gamma_*>0$ will be used as
additional piecewise constant control parameters.

Henceforth we consider systems with a single dominant Lindblad generator
(while small additional noise is treated in Appendix~E).
In the {\em non-unital case} it is the Lindblad generator for {\em amplitude damping},~$V_a$,
while in the {\em unital case} \footnote{
	A relaxation process $\dot{\rho} = -\Gamma(\rho)$ is {\em unital} if it preserves multiples of the
	identity, i.e.\  $\Gamma(\unity)= 0$, otherwise it is {\em non-unital}; here we first
	consider the amplitude-damping extreme case of non-unital noise, before generalising
	non-unital processes in~\cite[App.~B]{SuppMat}.
} 
it is the one for {\em bit flip},~$V_b$,  defined as
\begin{equation}\label{eqn:amp-damp+bit-flip}
V_a := \unity_2^{\otimes (n-1)}\otimes\ketbra{0}{1}\quad\text{and}\quad
V_b := \unity_2^{\otimes (n-1)}\otimes X/\sqrt{2}\;,
\end{equation}
where $X$, $Y$ and~$Z$ are the Pauli matrices.

As in~\cite{DHKS08}, we simply say a control system on $n$~qubits meets the condition for (weak)
Hamiltonian controllability if the Lie closure under commutation of its Hamiltonians comprises all unitary directions in the sense
\begin{equation}\label{eqn:closure-wh}
\hspace{-3mm}
\expt{iH_0, iH_j\,|\,j=1,\dots , m}_{\sf Lie} = \su(N)\;\text{with}\; N:=2^n .
\end{equation}
For the Lie-algebraic setting, see~\cite{JS72,SchiFuSol01,ZS11,ZS11add,Alt03,DHKS08}.
Now the {\em reachable set} $\reach_\Sigma(\rho_0)$ is defined as the set of all
states $\rho(\tauf)$ with $\tauf\geq 0$ that can be reached from $\rho_0$ following the dynamics of $\Sigma$.
If Eqn.~\eqref{eqn:closure-wh} holds,
without relaxation one can steer from any initial state $\rho_0$ to
any other state $\rho_{\rm target}$ with the same eigenvalues.
In other words, for $\gamma=0$ the control system $\Sigma$ acts transitively on the 
unitary orbit $\mathcal U(\rho_0):=\{U\rho_0 U^\dagger\,|\,U\in SU(N)\}$ 
of the respective initial state $\rho_0$. This holds for
any $\rho_0$ in the set of all density operators, termed $\pos_1$ henceforth.

Under coherent control and {\em constant noise} ($\gamma>0$ non-switchable) it is difficult to give precise
reachable sets for general $n$-qubit systems that satisfy Eqn.~\eqref{eqn:closure-wh}
only upon including the drift Hamiltonian~$H_0$~\cite{DHKS08,ODS11}.
Based on work by Uhlmann~\cite{Uhlm71, Uhlm72, Uhlm73},
majorisation criteria that are powerful if $H_0$ is not needed to meet  Eqn.~\eqref{eqn:closure-wh} 
\cite{Yuan10, Yuan11} now just give upper bounds to reachable sets by inclusions. 
Even worse, with increasing number of qubits $n$, these inclusions get increasingly inaccurate
and have to be replaced by Lie-semigroup methods as described in~\cite{ODS11}.

In the presence of {\em bang-bang switchable relaxation on a single
qubit} in an $n$-qubit system, here we show that the situation
improves significantly and one obtains 
two major results. Both are proven in the Supplement~\cite{SuppMat},
yet a {\em bird's-eye view} is added in the discussion below.

\begin{theorem}[non-unital]\label{thm:transitivity}
Let $\Sigma_a$ be an $n$-qubit bilinear control system as in Eqn.~\eqref{eqn:master}
satisfying Eqn.~\eqref{eqn:closure-wh} for $\gamma=0$.
Suppose the $n^{\rm th}$ qubit (say) undergoes 
amplitude-damping relaxation, the noise amplitude of which can be
switched in time between two values as $\gamma(t)\in\{0,\gamma_*\}$ with $\gamma_*>0$. 
If there are no further sources of decoherence, then 
the control system $\Sigma_a$
acts transitively on the set of all density operators $\pos_1$, i.e.
\begin{equation}
\overline{\reach_{\Sigma_a}^{\phantom{1}}(\rho_0)}=\pos_1 \quad\text{for all }\, \rho_0\in\pos_1\;,
\end{equation}
where the closure is understood as the limit $\gamma_* \tauf\to\infty$.
\end{theorem}

\begin{theorem}[unital]\label{thm:majorisation}
Let $\Sigma_b$ be an $n$-qubit bilinear control system as in Eqn.~\eqref{eqn:master}
satisfying Eqn.~\eqref{eqn:closure-wh} ($\gamma=0$)
now with the $n^{\rm th}$ qubit  undergoing 
bit-flip relaxation with switchable noise amplitude 
$\gamma(t)\in\{0,\gamma_*\}$.
If there are no further sources of decoherence, then 
(in the limit $\gamma_* \tauf\to\infty$) the reachable set to $\Sigma_b$
explores all density operators majorised by the initial state $\rho_0$, i.e.
\begin{equation}
\hspace{-3mm}
\overline{\reach_{\Sigma_b}^{\phantom{1}}(\rho_0)}=\{\rho\in\pos_1 \,|\,\rho\prec\rho_0\}\;\text{for any }\, \rho_0\in\pos_1\,.
\end{equation}
\end{theorem}
%
The proofs in Supplement~\cite[App.~A]{SuppMat} explicitly include possible Lamb shifts, and  App.~B
generalises the results to finite-temperature baths in order to go beyond
algorithmic cooling.

Thm.~1 is the first to show that for an any pair of states $(\rho_0,\rho_\text{target})$
connected via a {\em non}-Markovian Kraus map~$N$ 
(see footnote~\cite{Note1} for definition of Markovianity)
by $\rho_\text{target}=N(\rho_0)$, there always is a
time-dependent Markovian map $M$ made of coherent control with amp-damp noise switching
that 
takes the same initial state to the same target
by $\rho_\text{target}=M(\rho_0)$.
Yet, even close to the identity there are Kraus maps that cannot be obtained as
(necessarily Markovian) solutions of the Lindblad master equation~\cite{Wolf08a, DHKS08}.
For details see App.~G.

For implementation the main requirement is a fast switchable dominant noise source
on top of unitary controllability.
The preconditions for the Markovian setting of the Lindblad equations are well
approximated in experiments as soon as one has separate time scales for the system dynamics ($\tau_S$),
the coherent controls ($\tau_C$), the relaxation ($\tau_R$) and the bath correlation time ($\tau_B$):
The Born-Markov approximation holds if $\tau_R \gg \tau_B$, while
the secular approximation holds if $\tau_R \gg \tau_S$ and $\tau_C \gg \tau_S$,
as discussed extensively in App.~C.

Before suggesting an experimental implementation meeting these conditions for coherent control
extended by simplest noise switching in
a fast tunable-coupler qubit design called GMon as devised in the
Martinis group~\cite{Mart09, Mart13,Mart14}, we show basic features
in simple illustrative models.

\begin{figure}[!ht]
\hspace{10mm}{\sf (a)}\hspace{34mm}\sf{(b)} $\hfill$\\[-0mm]
\includegraphics[width=0.43\columnwidth]{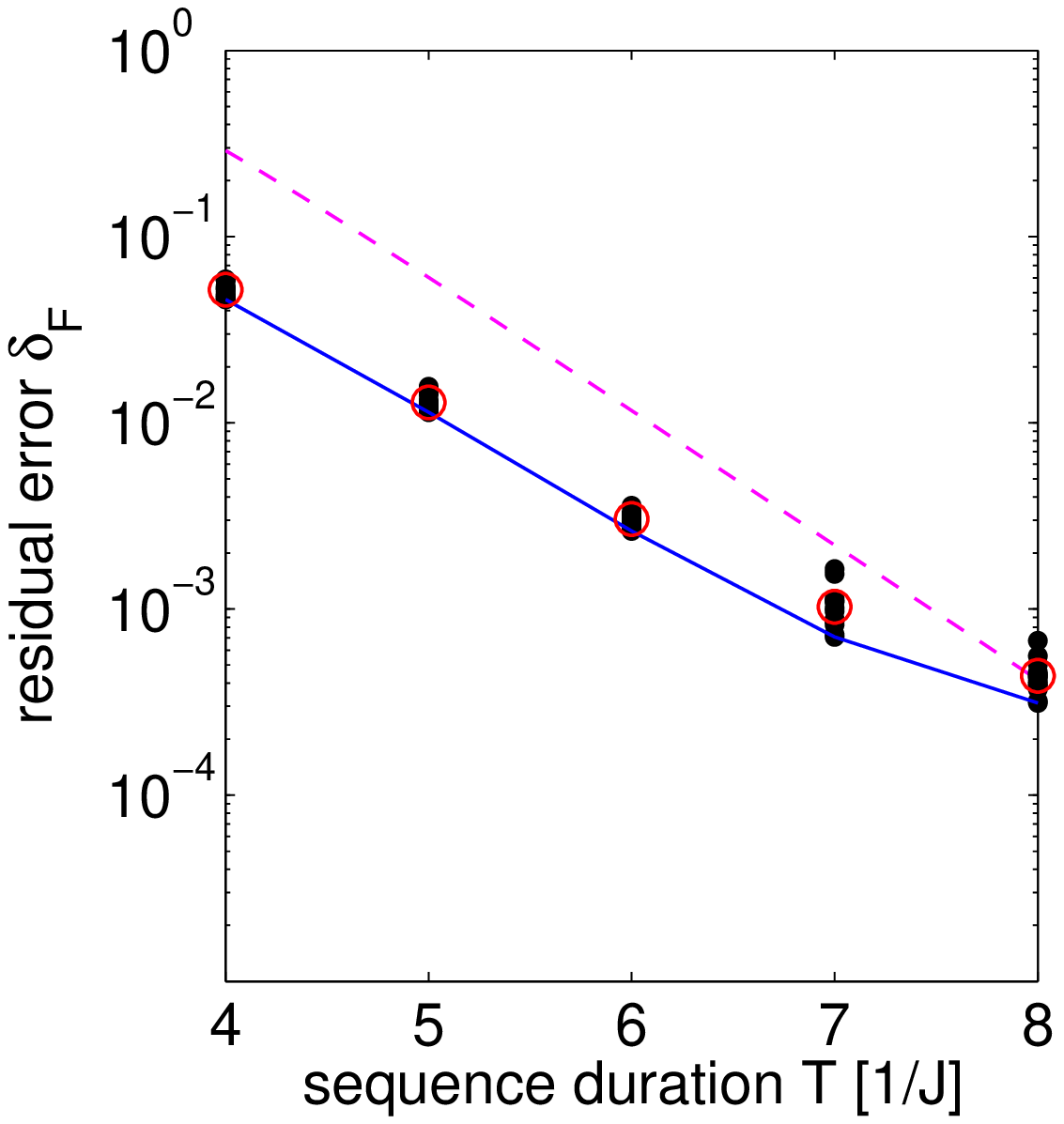}
\includegraphics[width=0.43\columnwidth]{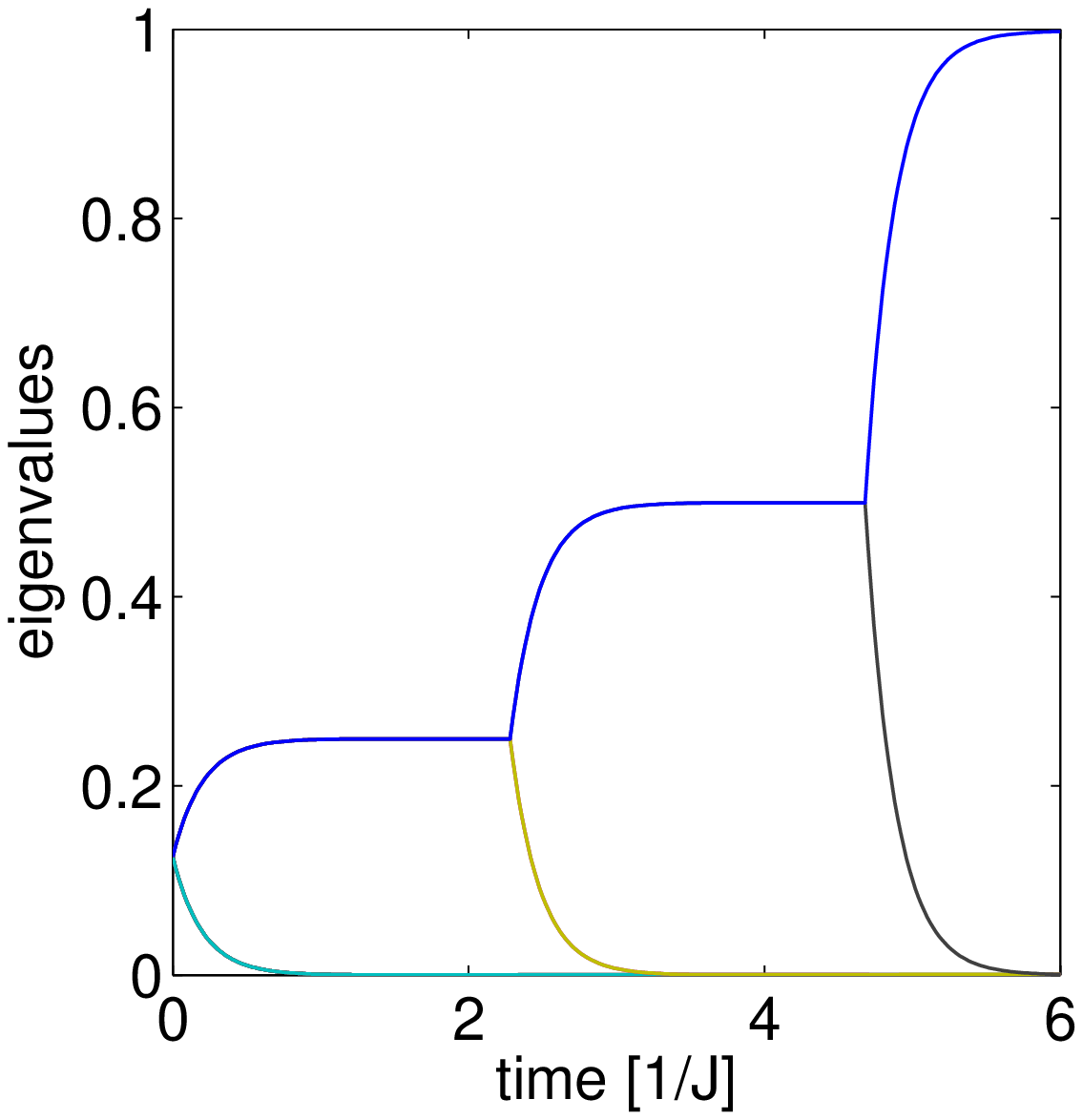}\\
\hspace{10mm}{\sf (c)}$\hfill$\\
\includegraphics[width=0.86\columnwidth]{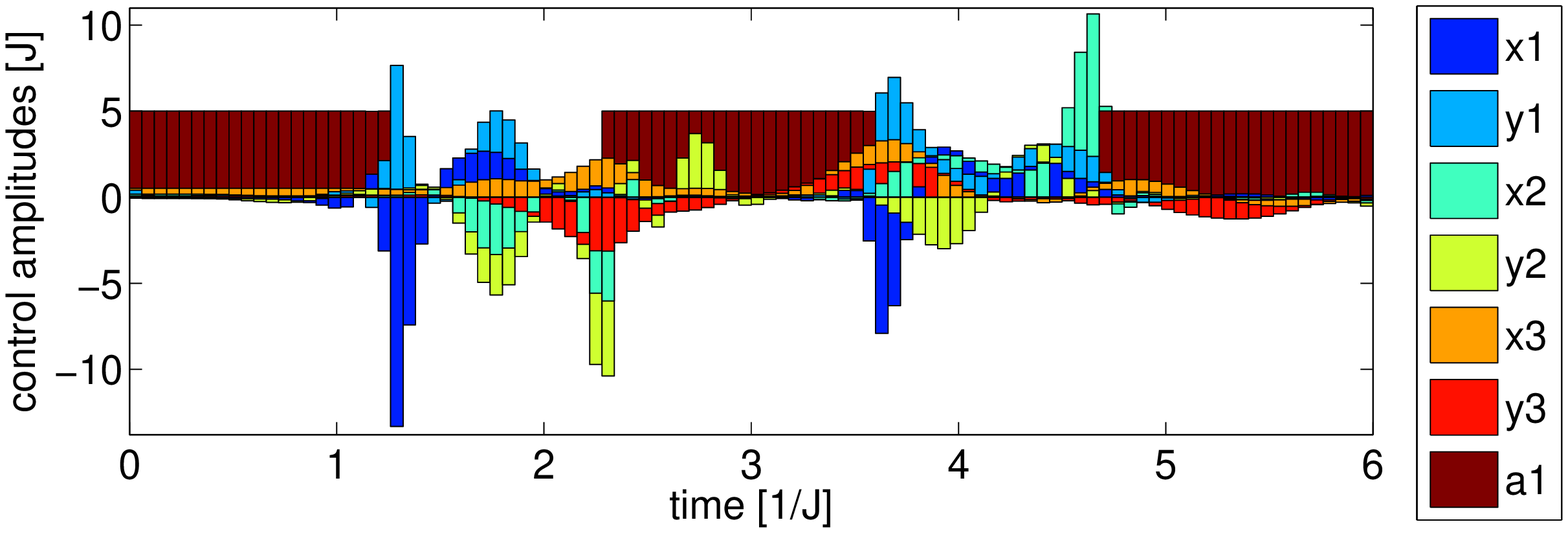}
\caption{\label{fig:th0}
Cooling the max.~mixed state~$\rho_\text{th}=\frac{1}{8}\unity$ to the ground state~$\ket{000}$
in a \mbox{3-qubit} Ising-{\sc zz} chain
with controlled amplitude-damping noise on qubit one
as in {\bf Example~1}.
(a)~Frobenius-norm error versus total sequence duration~$\tauf$.
The dashed line gives the upper bound from Eqn.~\eqref{eq:time-ex1},
and the dots (red circles for averages) individual optimization runs with 
random initial sequences.
(b)~Evolution of the eigenvalues under the best of the $\tauf=6/J$
solutions.
This sequence (c) shows three relaxative periods with maximal noise
amplitude~$\gamma_{a1}$ for transforming eigenvalues together with
unitary actions. 
Each purely unitary segment is of the approximate duration~$1/J$,
corresponding to the duration of a single \iSWAP.
}
\end{figure}
%
{\em Explorative Model Systems.}
To challenge our optimal control algorithm,
we first consider two examples of state transfer where the target states
can only be reached {\em asymptotically} ($\gamma_* \tauf\to\infty$), i.e.\ they are
in the closure of the 
reachable sets.
To illustrate Thms.~\ref{thm:transitivity} and~\ref{thm:majorisation},
next we
show noise-driven transfer (i)~between \emph{random pairs of states} under controlled
amplitude damping 
and (ii)~between random pairs of states satisfying $\rho_{\rm target}\prec\rho_0$
under controlled bit-flip noise in Examples 3 and~4.

In Examples 1--4, our system is an $n$-qubit chain with uniform
\mbox{Ising-{\sc zz}} nearest-neighbour couplings given by
$H_0 := \pi J \sum_k \tfrac{1}{2} Z\sys{k} Z\sys{k+1}$,
and piecewise constant $x$ and~$y$ controls (that need not be bounded) on each qubit
locally, so the control system satisfies Eqn.~\eqref{eqn:closure-wh}.
We add controllable noise (amplitude-damping or
bit flip) with amplitude $\gamma(t)\in[0,\gamma_*]$
and~$\gamma_* = 5 J$
acting on one terminal qubit.
The control system is detailed in~\cite{SuppMat}.

\begin{example}

Here, as for initialising a quantum computer~\cite{VincCriteria},
the task is to turn the high-$T$ initial state
$\rho_{\text{th}}:=\tfrac{1}{2^n}\unity$ into the pure target state
$\ket{00\ldots0}$
by unitary control and controlled amplitude damping.
For $n$ qubits, the task can be accomplished in an \mbox{$n$-step} protocol: 
let the noise act on each qubit~$q$ for the time~$\tau_q$ to 
populate the state $\ketbra{0}{0}$,
and permute the qubits between the steps.
For a linear chain this requires
$\sum_{q=1}^n (q-1)=\binom{n}{2}$ 
nearest-neighbour \SWAP s.
Since all the intermediate states are diagonal, the \SWAP s can be
replaced with \iSWAP s, each taking a time of~$\frac{1}{J}$
under the Ising-{\sc zz} coupling.
The residual Frobenius-norm error~$\delta_F$
is minimised when all the $\tau_q$ are equal, giving
$\delta_{F_a}^2(\expfactoramp)
= 1 -2 \left(1-\tfrac{\expfactoramp}{2}\right)^n +\left(1 -\expfactoramp +\tfrac{1}{2} \expfactoramp^2\right)^n$,
where $\expfactoramp:=e^{-\gamma_* \tau_q}$.
Linearizing
at $\expfactoramp = 0$
and adding time for the \iSWAP{}s, the total duration
of this simple protocol as a function of~$\delta_F$
amounts (to first order in $\expfactoramp$) to
\begin{equation}\label{eq:time-ex1}
\tauf_a \approx \binom{n}{2}\tfrac{1}{J} +\tfrac{n}{\gamma_*} 
	\ln \Big(  \tfrac{\sqrt{n(n+1)}}{2 \delta_{F}^{\phantom{|}} } \Big).
\end{equation}

Fig.~\ref{fig:th0} shows that optimal control can outperform
this simple scheme by parallelising 
unitary
transfer with the amplitude-damping driven \/`cooling\/' steps.
Moreover, initialisation can still be accomplished to a good
approximation  if unavoidable constant dephasing noise on all the
three qubits is present, as shown in App.~E,
while App.~B shows how combining coherent control concomitant with incoherent
control goes beyond algorithmic cooling on a general scale.
It is a generic example of shorter and more efficient control sequences than obtained by
conventional
separation of (algorithmic) cooling and processing.

\end{example}

\begin{example}

In turn, consider \/`erasing\/' the pure initial
state~$\ket{00\ldots0}$
to the high-$T$ state~$\rho_{\text{th}}$.
Under controlled amplitude damping
this can be accomplished exactly, each round splitting the populations in half with a total
time of
$
\tauf'_a = \binom{n}{2}\tfrac{1}{J} + \tfrac{n}{\gamma_*} \ln(2).
$
Yet with {\em bit-flip noise} this transfer
can only be obtained asymptotically.
One may use a similar $n$-step protocol as in the previous example,
this time approximately erasing each qubit to a state proportional
to~$\unity$.
Again, optimal control 
is much faster than
this simple scheme
(results and details in~App.~E).

\begin{figure}[!ht]
\hspace{2mm}{\sf (a)}\hspace{25mm}{\sf (b)}\hspace{25mm}{\sf(c)}$\hfill$\\
\raisebox{-.6mm}{\includegraphics[width=0.327\columnwidth]{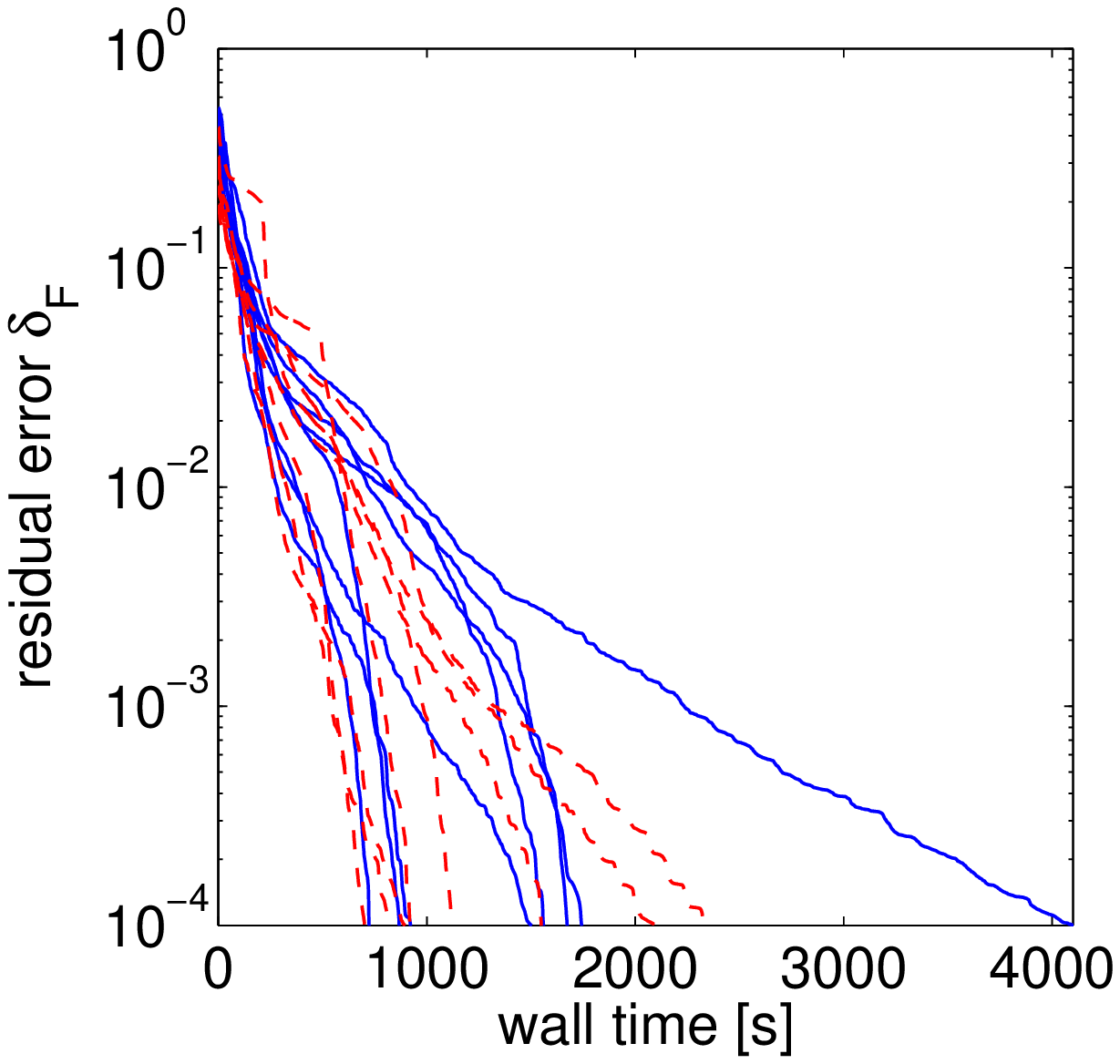}}
\includegraphics[width=0.32\columnwidth]{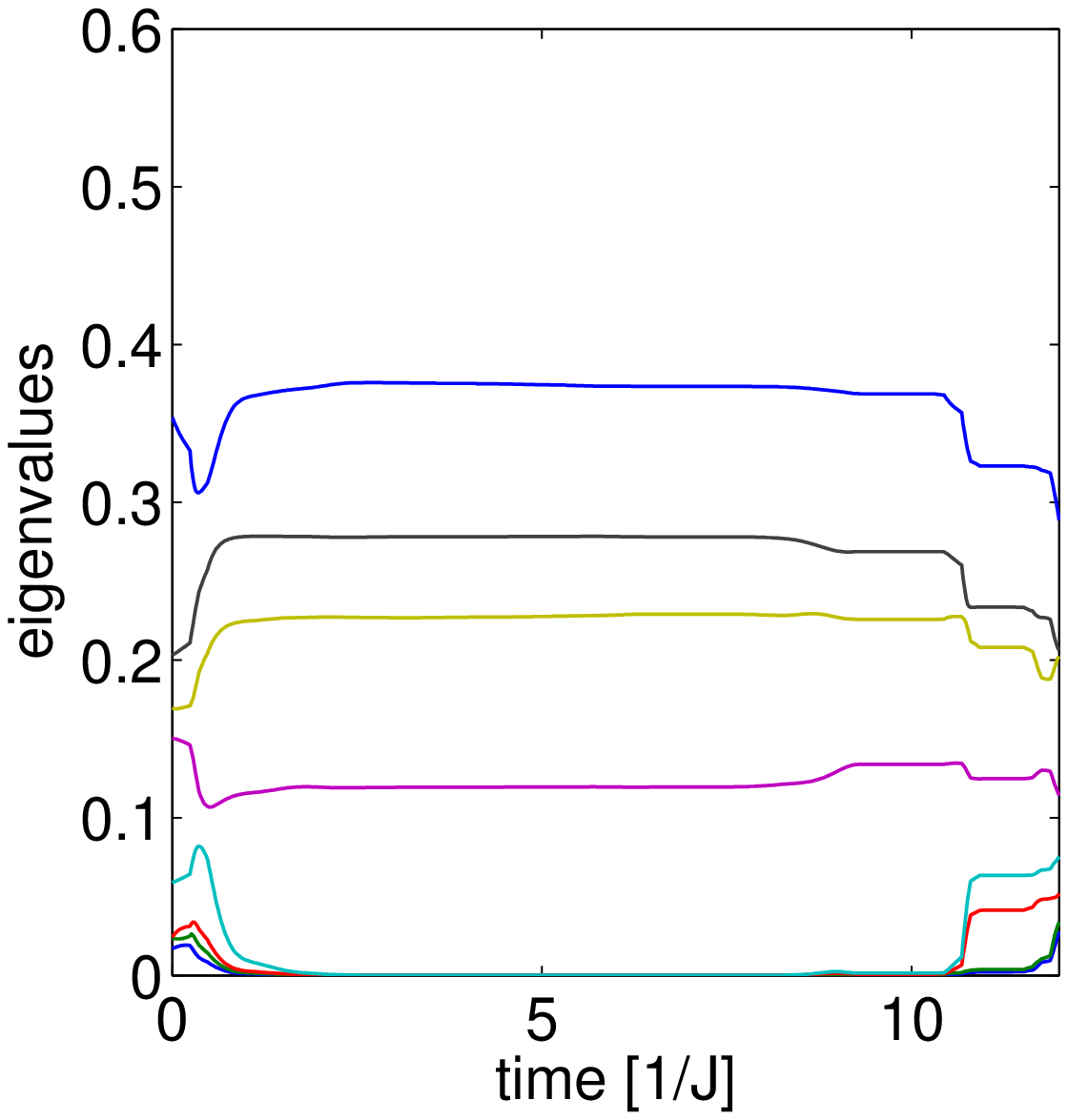}
\includegraphics[width=0.32\columnwidth]{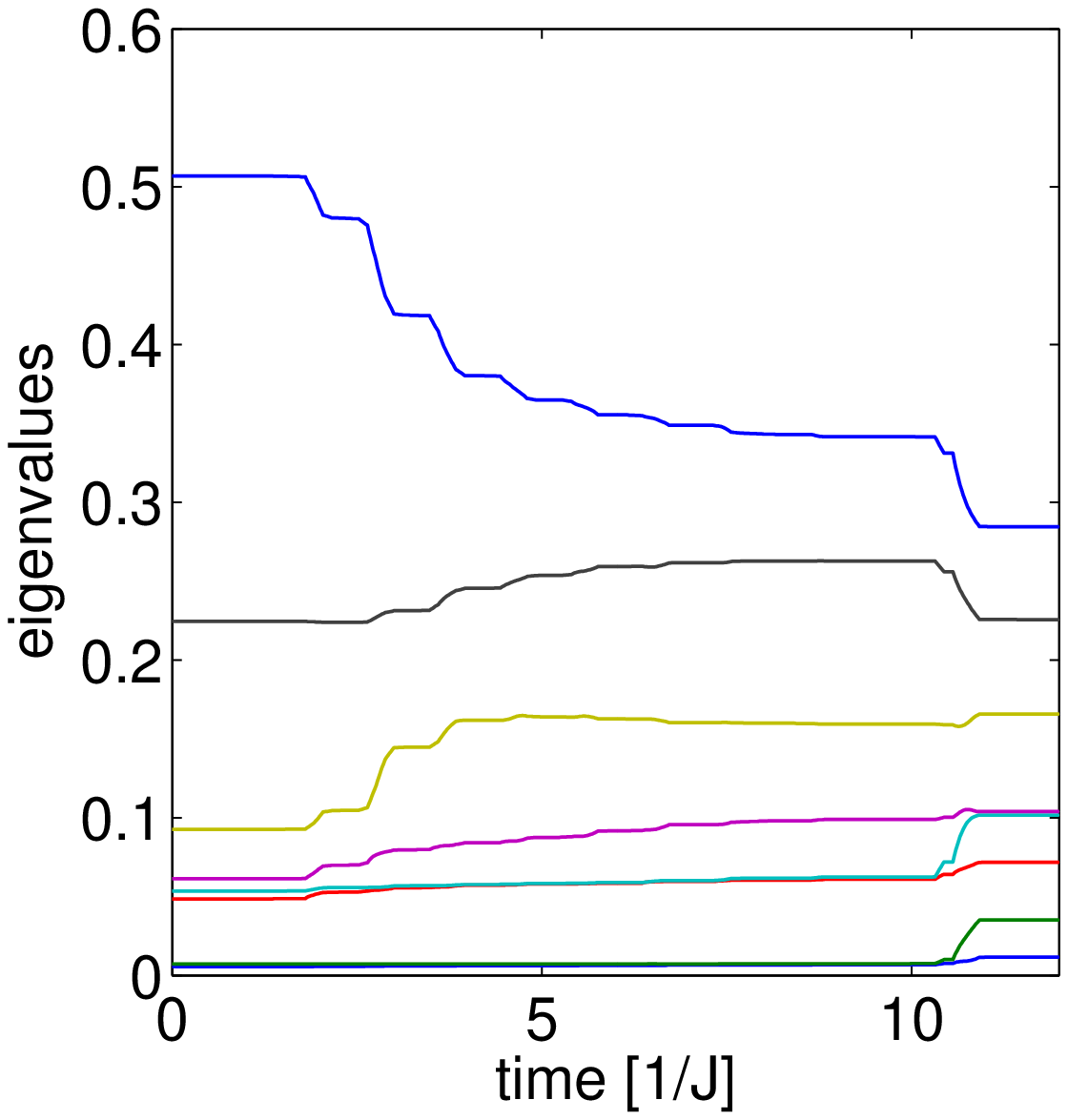}
\caption{\label{fig:rand}
(a) Error vs.~computation time for state transfer between
pairs ($\rho_0,\rho_{\rm target}$)
of random \mbox{3-qubit} states using controlled amplitude-damping
noise (solid) in addition to local unitary control. 
Same for random pairs  ($\rho_0,\rho_{\rm target}$)
with $\rho_{\rm target}\prec\rho_0$
under controlled bit-flip noise (dashed).
In both cases (a) shows the median of 9
optimisation runs for each of the eight random state pairs.
Representative examples of evolution of the eigenvalues for an amplitude-damping transfer (b) and
for a bit-flip transfer (c). In the former case, a typical feature is the
initial zeroing of the smaller half of the eigenvalues while the larger
half are re-distributed among themselves. Only at the very end the
smaller eigenvalues revive.
}
\end{figure}
\end{example}

\begin{example}
We illustrate transitivity under
controlled amplitude damping on one qubit plus general unitary control
by {\em transfers between pairs of random} \mbox{3-qubit} density operators.
Fig.~\ref{fig:rand}(a) shows results well within~$\delta_F = 10^{-4}$.
\end{example}

\begin{example}
Similarly, with controlled bit-flip noise on one qubit plus
general unitary control,
Fig.~\ref{fig:rand} 
illustrates transfer between pairs of random 3-qubit states 
solely constrained by the unitality condition  $\rho_{\rm target}\prec\rho_0$.
\end{example}

In the final formal {\bf Example~5}, App.~E shows
how
to drive a system of  four trapped ion qubits 
from the high-$T$ initial state
$\rho_{\text{th}}:=\tfrac{1}{2^n}\unity$ to the pure entangled target state
$\ket{\text{\sc ghz}_4} = \tfrac{1}{\sqrt{2}}(\ket{0000}+\ket{1111})$.
In contrast to~\cite{BZB11}, where an ancilla qubit was
added  (following~\cite{VioLloyd01}) for a {\em measurement-based
circuit on the $4+1$ system}, one can do {\em without ancilla qubit} by
controlled amplitude-damping just on the terminal qubit together with coherent control.

\medskip

{\em Proposed Experimental Implementation by GMons.}
%

Superconducting charge and flux qubits have passed various iterations of designs
leading to transmons~\cite{Koch_2007, Schreier_2008}, i.e.\
weakly anharmonic oscillators whose energy spectrum is insensitive
to slow charge noise.
For circuit {\sc qed},
the coupling element is a spatially distributed resonator and qubits
are frequency-tuned relative to it. Tunable couplers were implemented
for flux qubits~\cite{Mooji_1999,Orlando_1999,Makhlin_1999,Makhlin_2001,Plourde_2004,Hime_2006}.

The {\em fast} tunable-coupler-qubit design devised in the Martinis group is 
called GMon~\cite{Mart09, Mart13,Mart14}.
It allows the implementation of tunable couplings
with similar parameters between qubits and between transmission lines \cite{Mart14}
(one of them open).
This can straightforwardly be extended to the case of coupling a qubit to a line.
Recently, the GMon has solved a lot of technological challenges, 
rendering it an effective tunable-coupling strategy between qubits and resonators 
(see also Refs.~\cite{Hoffman_2011, Srinivasan_2011}).

In App.~C we give a detailed derivation how
this fast tunable coupling to an open transmission line can be used to experimentally
implement controllable Markovian amplitude damping noise (with an effective Boltzmann factor of~$10^{-3}$).
We discuss the weak-coupling and singular-coupling limits and describe
scenarios ensuring the Lamb-shift Hamiltonian term induced by switching on the
noise
does not compromise Thms.~1 and 2.

Fig.~\ref{fig:GMon-Siewert} shows how switchable noise, in parallel with coherent controls,
produces a {\sc ppt}-entangled state~\cite{Siewert2016}
on 
two coupled GMon qutrits.
In App.~C
we add further numerical results to show how initialisation, 
erasure, and preparing a {\sc ghz}-type state can be implemented likewise.

\begin{figure}[!ht]
\hspace{7mm}{\sf (a)}$\hfill$\\
\includegraphics[width=.85\columnwidth]{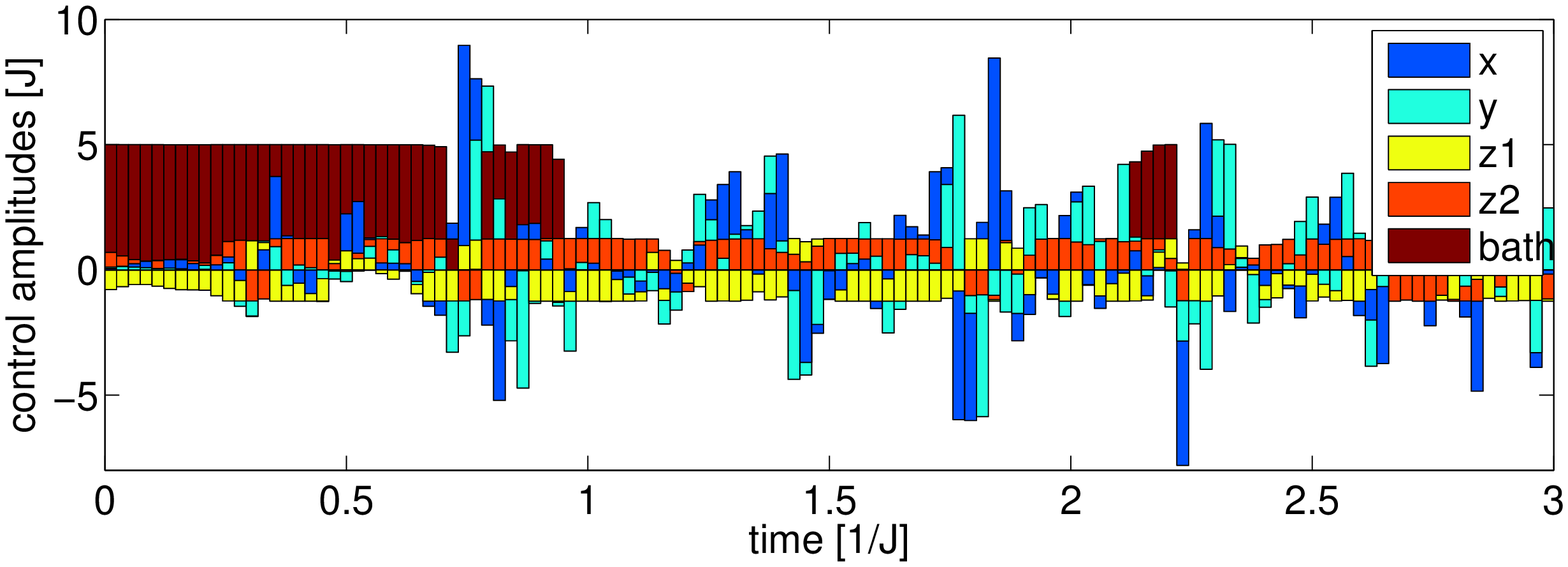}\\
\hspace{7mm}{\sf (b)}$\hfill$\\
\includegraphics[width=.85\columnwidth]{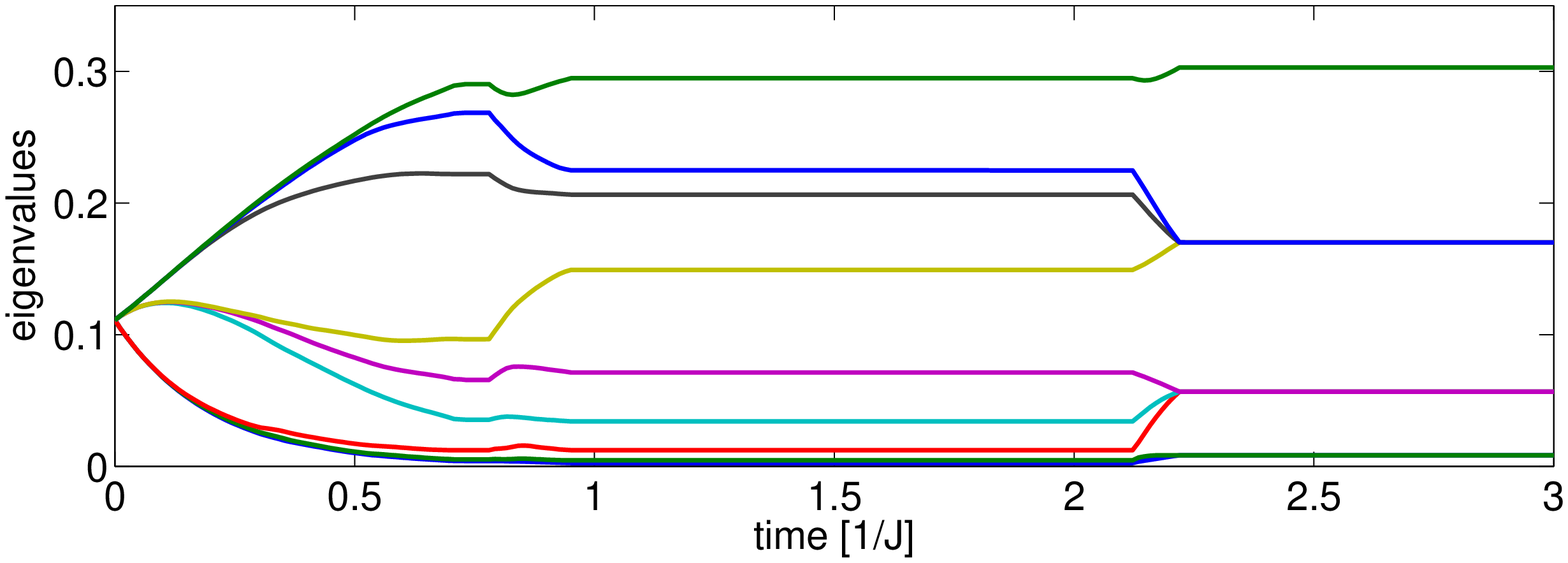}\\
\hspace{7mm}{\sf (c)}$\hfill$\\[-2mm]
\includegraphics[width=.6\columnwidth]{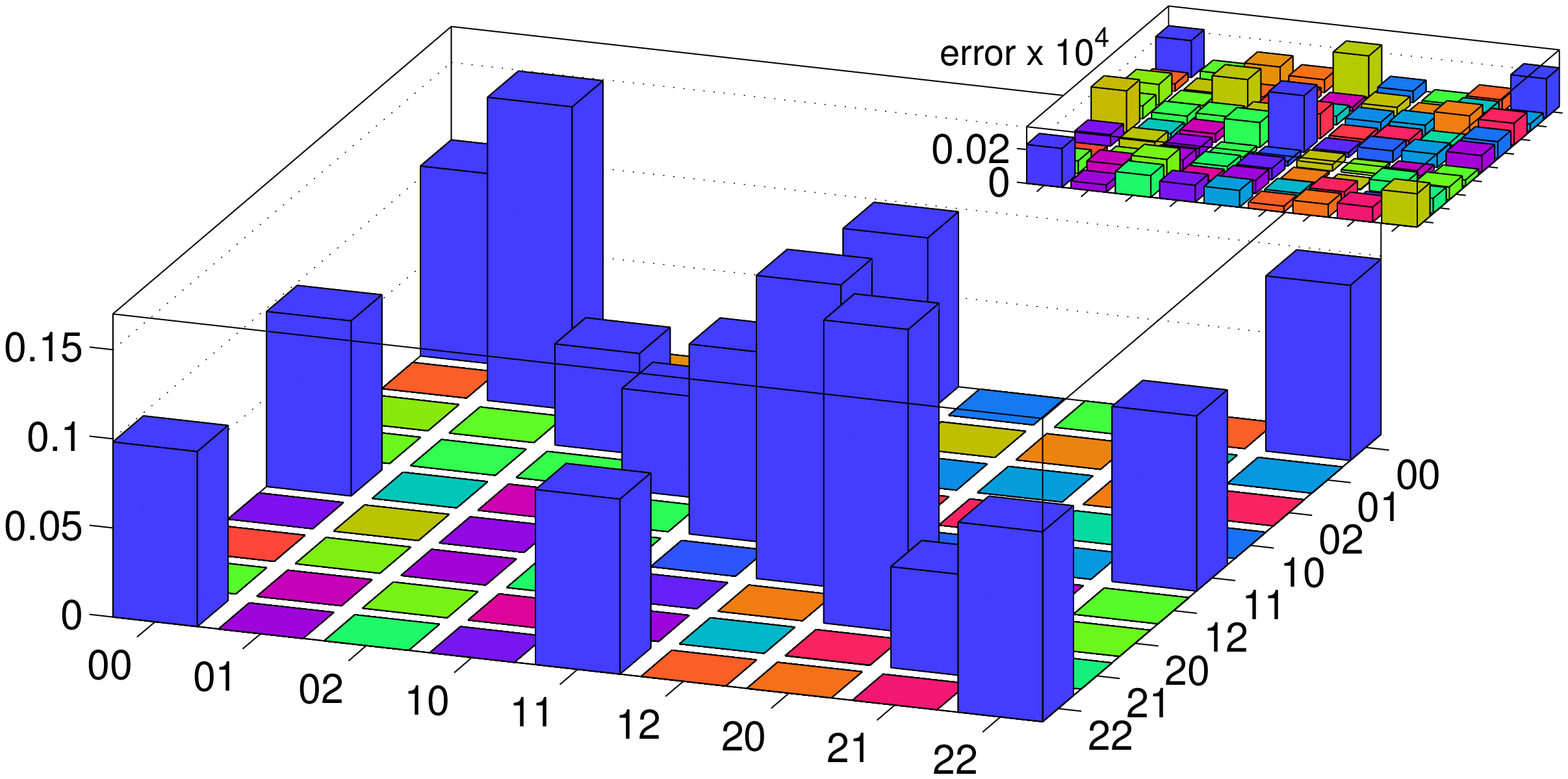}
\caption{\label{fig:GMon-Siewert}
Transfer from the maximally
mixed state~$\rho_\text{th} = \tfrac{1}{9}\unity$ to a {\sc ppt}-entangled 2-qutrit state~\cite{Siewert2016,SES16}
in a device of two coupled GMons~\cite{Mart09, Mart13,Mart14} with a
geometrized $J$~coupling of $160$~MHz.
Switchable amplitude-damping noise is implemented by a
tunable coupling to an open transmission line at~$35$~mK.
(a) Optimized sequence of coherent controls and bath couplings.
Coherent controls amount to joint $x$ and $y$ pulses on both GMon qutrits,
while $z$ pulses can be performed individually.
The sequence shows noise controls
in parallel to the unitary actions.
The residual error $\delta_F<10^{-4}$.
(b) Eigenvalue evolution under the sequence.
(c) Tomography of the final state. The inset shows the
error multiplied by a factor of~$10^4$.
}
\end{figure}

{\em Discussion.}
In bird's-eye view, our scheme may be described as follows:
By unitary controllability, we may diagonalise the initial and the target states.
So transferring a diagonal initial state into a diagonal target state 
can be considered as the {\em normal form of the state-transfer problem}. 
This form can be treated analytically, because
it is easy to separate dissipation-driven {\em changes of eigenvalues} from unitary coherent
actions of {\em permuting eigenvalues} and decoupling drift Hamiltonians,
while numerical controls may profit from doing both in parallel.



One may contrast our method with the closed-loop control method
in~\cite{VioLloyd01} originally designed for quantum-map synthesis
using projective measurement of a coupled resettable ancilla qubit
plus full unitary control to enact
arbitrary quantum operations (including state transfers),
with Markovian evolution as the infinitesimal limit.
Applied to state transfer, the present method instead relies on a switchable local Markovian noise source
and requires {\em neither measurement nor an ancilla}~\footnote{
	The Supplement~\cite[App.~G]{SuppMat} explains how
	for {\em state transfer}, Markovian open-loop controllability already implies full
	state controllability (including transfers by non-Markovian processes),
	while for the lift to {\em Kraus-map controllability} it remains an open question whether
	closed-loop feedback control is not only sufficient (as established in \cite{VioLloyd01}), but
	also necessary in the sense that it could not be replaced by open-loop unitary control plus
	control over some local noise incorporated in a {\em non-Markovian master equation}
	of the recent type of~Refs.~\cite{DiosiFerialdi_NMarkovMasterEqn2014,
	Ferialdi_NMarkovMasterEqn2016, Vacc_NMarkovMasterEqn2016,
	Ferialdi_NMarkovSpinBosonJayCumm2017}.
}.\\[-2mm]

{\em Outlook.}
We have proven that by adding as a new control parameter 
bang-bang switchable Markovian noise on just one system qubit,
an otherwise coherently controllable $n$-qubit network can explore unprecedented
reachable sets: in the case of amplitude-damping noise
(or any noise process in its unitary equivalence class)
one can convert
{\em any} initial state $\rho_0$ into {\em any} target state $\rho_{\rm target}$,
while under switchable bit-flip noise
(or any noise process unitarily equivalent)
one can transfer any~$\rho_0$ into
any target~$\rho_{\rm target}\prec\rho_0$ {\em majorised by the initial state}.
These results have been further generalised and compared to equilibrating the system with
a finite-temperature bath. 

To our knowledge, this is the first time these features are systematically explored
as {\em open-loop} control problems and solved 
in a minimal setting
by coherent  local controls 
and bang-bang modulation of a single {\em local}  Markovian noise source.
For {\em state transfer}, our open-loop Markovian protocol ensures full state controllability,
and is as powerful as the closed-loop
measurement-based feedback scheme in \cite{VioLloyd01},
so it may simplify many experimental implementations.\\

{\em Conclusions.}
We have extended quantum optimal control platforms like {\sc dynamo}~\cite{PRA11} by controls
over Markovian noise amplitudes in otherwise coherently controllable systems.
We exemplified
initialisation 
to the pure zero-state, state erasure, and the interconversion of
random pairs of mixed states.
For finite temperatures, we showed that combining coherent and incoherent controls
supersedes algorithmic cooling~\cite{SMW05} and combines cooling with simultaneous unitary processing.
In a detailed worked example, we propose using explicit recent GMon settings~\cite{Mart13,Mart14}
for experimental implementation of switchable noise by fast tunable coupling to an
open transmission line together with coherent controls.

We thus anticipate that combining coherent with simplest incoherent controls
paves the way to many other unprecedented applications of
state transfer and quantum simulation exploring the limits of Markovian dynamics.

\begin{acknowledgments}
We wish to thank Daniel Lidar, Lorenza Viola, 
Jens Siewert, 
and Alexander Pechen for useful comments
mainly on the relation to their works.
This research was supported in part by the {\sc eu} projects {\sc siqs} and {\sc quaint},
exchange with {\sc coquit},
by {\em Deutsche Forschungsgemeinschaft} in {\sc sfb}~631 and
{\sc for}~1482,
and by the excellence network of Bavaria ({\sc enb}) through {\sc exqm}.
\end{acknowledgments}

\bibliography{control21ville}
%

\clearpage
\onecolumngrid
\tableofcontents
\clearpage

\appendix


\noindent

\section{Proofs of the Main Theorems and Generalising Remarks}\label{sec:proofs}
\setcounter{theorem}{0}
\setcounter{figure}{3}
For clarity of arguments, first we prove Theorems~1 and 2 of the main text under the
(unnecessary) simplifying assumption of diagonal drift and Lamb-shift Hamiltonians $H_0+H_{LS}$.
A simple Trotter argument proven in Corollary~\ref{cor:Trotter-Decoup} below then recovers
the stronger version for arbitrary forms of~$H_0$ and $H_{LS}$ given in the main text.

\begin{theorem}[non-unital extreme case]
\label{Athm:transitivity}
Let $\Sigma_a$ be an $n$-qubit bilinear control system as in Eqn.~(1) of the main text 
satisfying Eqn.~(4) 
for $\gamma=0$.
Suppose the $n^{\rm th}$ qubit (say) undergoes
amplitude-damping relaxation by $V_a:=\left(\begin{smallmatrix} 0 & 1\\ 0 & 0\end{smallmatrix}\right)$, 
the noise amplitude of which can be
switched in time between two values as $\gamma(t)\in\{0,\gamma\}$ with $\gamma>0$. 
If the free evolution Hamiltonian~$H_0$, e.g.,\ of Ising-$ZZ$ type
(and the Lamb-shift term $\Hlamb$) are diagonal,
and if there are no further sources of decoherence, then 
$\Sigma_a$
acts transitively on the set of all density operators $\pos_1$ in the sense
\begin{equation}
\overline{\reach_{\Sigma_a}^{\phantom{1}}(\rho_0)}=\pos_1 \quad\text{for all }\, \rho_0\in\pos_1\;,
\end{equation}
where the closure is understood as the limit $\gamma \tauf\to\infty$
and $\tauf$ is the duration of the control sequence.

\begin{proof}
We keep the proofs largely constructive. By unitary controllability $\rho_0$
may be made diagonal, so the vector of diagonal elements is $r_0~:=~\spec(\rho_0)$.
Since a diagonal $\rho$ commutes with a diagonal $H_{0}$, the
state remains stationary under the free evolution.
The evolution of the vector of diagonal elements under the noise follows
\begin{equation}\label{eqn:thm1}
r(t) =
R_a(t)\, r_0 :=
\left[\unity_2^{\otimes(n-1)}\otimes\begin{pmatrix}
		1 & 1-\expfactoramp\\
		0 & \expfactoramp
\end{pmatrix}\right]  r_0,
\end{equation}
where $\expfactoramp := e^{-\gamma t}$ and $R_a(t)$ is by
construction a stochastic matrix.
With the noise switched off, full unitary control includes arbitrary permutations
of the diagonal elements. Any of the pairwise relaxative transfers between
diagonal elements $\rho_{ii}$ and $\rho_{jj}$  (with $i\neq j$)
lasting a total time of $\tau$ can be neutralised by permuting $\rho_{ii}$ and $\rho_{jj}$ 
after a time
\begin{equation}\label{eqn:switch-time1}
\tau_{ij}:=
\tfrac{1}{\gamma} \ln \left(\frac{(\rho_{ii}/\rho_{jj})\,e^{\gamma \tau}+1}{(\rho_{ii}/\rho_{jj})+1}\right)
\end{equation}
and letting the system evolve under noise again for the remaining time $\tau-\tau_{ij}$.
Thus with $2^{n-1}-1$ such switches all but the one desired transfer can
be neutralised.
As $\rho(t)$ remains diagonal under all permutations, 
relaxative and coupling processes, one can obtain any state of the form
\begin{equation}
\rho(t)=
\diag(\ldots, \: {[\rho_{ii} + \rho_{jj}\cdot (1-e^{-\gamma t})]}_{ii},
 \: \ldots, {[\rho_{jj}\cdot e^{-\gamma t}]}_{jj}, \: \ldots).
\end{equation}
Sequences of such transfers between single pairs of eigenvalues 
$\rho_{ii}$ and $\rho_{jj}$  and their permutations then generate (for
$\gamma \tauf\to \infty$) the entire set of all 
diagonal density operators $\Delta\subset\pos_1$. 
By unitary controllability one gets all the unitary orbits $\mathcal U(\Delta)=\pos_1$.
Hence the result.
\end{proof}
\end{theorem}

\begin{theorem}[unital case]
\label{Athm:majorisation}
Let $\Sigma_b$ be an $n$-qubit bilinear control system as in Eqn.~(1) of the main text 
satisfying Eqn.~(4) 
now with the $n^{\rm th}$ qubit (say) undergoing bit-flip relaxation with
$V_b:=\tfrac{1}{\sqrt{2}}\left(\begin{smallmatrix} 0 & 1\\ 1 & 0\end{smallmatrix}\right)$
with switchable noise amplitude 
$\gamma(t)\in\{0,\gamma\}$.
If the free evolution Hamiltonian~$H_0$, e.g.,\ of Ising-$ZZ$ type
(and the Lamb-shift term $\Hlamb$) are diagonal,
and if there are no further sources of decoherence, then 
in the limit $\gamma \tauf\to\infty$ (with $\tau$ as duration of the
control sequence) the reachable set to $\Sigma_b$
explores all density operators majorised by the initial state $\rho_0$ in the sense
\begin{equation}
\overline{\reach_{\Sigma_b}^{\phantom{1}}(\rho_0)}=\{\rho\in\pos_1 \,|\,\rho\prec\rho_0\}\;\text{for any }\, \rho_0\in\pos_1\,.
\end{equation}

\begin{proof}
Again start by unitarily diagonalizing the initial state, with
$r_0 := \spec(\rho_0)$.
The evolution under the noise remains diagonal following
\begin{equation}\label{eqn:thm2}
r(t) =
R_b(t) \, r_0 :=
\left[\unity_2^{\otimes(n-1)}\otimes\tfrac{1}{2}\begin{pmatrix}
		(1 + \expfactorbitflip)& (1-\expfactorbitflip)  \\
		(1-\expfactorbitflip) &  (1 + \expfactorbitflip)
\end{pmatrix}\right] r_0,
\end{equation}
where $\expfactorbitflip:= e^{-\gamma t}$ and $R_b(t)$ is doubly stochastic.
In order to limit the relaxative averaging to the first two eigenvalues,
first conjugate the diagonal~$\rho$ with the unitary
\begin{equation}\label{eqn:protect}
U_{12}:= \unity_2 \oplus
R^{\oplus 2^{n-1}-1},
\end{equation}
where $R = e^{-\frac{\pi}{2} Y/2}$ is a $\pi/2$ rotation around the $y$~axis,
to obtain $\rho':=U_{12}\rho U^\dagger_{12}$.
Then the relaxation acts as a {$T$-transform}~\footnote{
	A $T$-transform is a convex combination  
	$\lambda\unity+(1-\lambda) Q$, where $Q$~is a pair transposition
	matrix and $\lambda\in [0,1]$.}
on the first two eigenvalues of $\rho'$, while leaving the others invariant.

Yet the protected subspaces have to be decoupled from the
free evolution Hamiltonian $H_0$ and the Lamb-shift term $\Hlamb$ (both assumed diagonal for the moment).
Any diagonal $H_0' := H_0+\Hlamb$
decomposes as
$H_0' = H_{0,1} \otimes \unity_2 +H_{0,2} \otimes Z$,
where $H_{0,1}$ and $H_{0,2}$ are again diagonal.
The first term commutes with~$\rho'$
and can thus be neglected.
The second
can be sign-inverted
by $\pi$-pulses in $x$-direction on the noisy qubit,
which also leave the bit-flip noise generator invariant:
\begin{equation}
(\unity \otimes e^{i\pi X/2}) e^{-t(\Gamma+i\hat H_{0,2}\otimes Z)} (\unity \otimes e^{-i\pi X/2})
= e^{-t(\Gamma-i\hat H_{0,2}\otimes Z)}.
\end{equation}
Thus $H_0$ and the Lamb-shift term $\Hlamb$
may be fully decoupled in the Trotter limit (NB: as shown in Corollary~\ref{cor:Trotter-Decoup} below,
a similar argument holds for $H_0+H_{LS}$ of arbitrary form)
\begin{equation}\label{eqn:Trotter-dec}
\lim_{k\to\infty}\;(e^{-\tfrac{t}{2k}(\Gamma+i\hat H_{0,2}\otimes Z)}  e^{-\tfrac{t}{2k}(\Gamma-i\hat H_{0,2}\otimes Z)})^k =  e^{-t\Gamma}\;.
\end{equation}
By combining permutations of diagonal elements
with selective pairwise averaging by relaxation, any \mbox{$T$-transform} of~$\rho$~\footnote{
	$R_b(t)$ of Eqn.~\eqref{eqn:thm2} covers $\lambda\in [\tfrac{1}{2},1]$, while
	$\lambda\in [0,\tfrac{1}{2}]$ is obtained by unitarily swapping the elements
	before applying $R_b(t)$;
	$\lambda=\tfrac{1}{2}$ is obtained in the limit 
	$\gamma \tauf \to \infty$. 
	}
can be obtained in the limit $\gamma \tauf\to\infty$:
\begin{align}\label{eqn:T-trafo}
\notag
\rho(t) =
\diag\big(\ldots, \: &\tfrac{1}{2}[\rho_{ii} + \rho_{jj} + (\rho_{ii} - \rho_{jj})\cdot e^{-\gamma t}]_{ii}, \: \ldots,\\
&\tfrac{1}{2}[\rho_{ii} + \rho_{jj} + (\rho_{jj}-\rho_{ii}) \cdot e^{-\gamma t}]_{jj}, \: \dots \big).
\end{align}

Now recall that a vector $y\in\mathbb R^N$ majorises a vector 
$x\in\mathbb R^N$, denoted $x\prec y$,
if and only if there is a doubly stochastic matrix $D$ with $x=D y$, where $D$ has to be a product of at most $N-1$ such
$T$-transforms (e.g., Thm.~B.6 in~\cite{MarshallOlkin} or Thm.~II.1.10 in~\cite{Bhatia}).
The decomposition into $T$-transforms is essential, because {\em not every} doubly stochastic matrix 
can be written as a product of $T$-transforms \cite{MarshallOlkin}.
Actually, by the work of Hardy, Littlewood, and P{\'o}lya~\cite{HLP34} this sequence of 
\mbox{$T$-transforms} is constructive~\cite[p32]{MarshallOlkin} as will be made use of later.
Thus in the limit $\gamma \tauf\to\infty$
all vectors of eigenvalues $r\prec r_0$ can be reached and hence by unitary controllability
all the states $\rho\prec\rho_0$.

Finally, to see that one cannot go beyond the states majorised by the initial state,
observe that controlled unitary dynamics combined with bit-flip relaxation is still completely positive,
trace-preserving and {\em unital}. Thus it takes the generalised form of a {\em doubly-stochastic linear map} 
$\Phi$ in the sense of~Thm.~7.1 in~\cite{Ando89}, which for any hermitian matrix~$A$ ensures 
$\Phi(A)\prec A$. Hence (the closure of) the reachable set is indeed confined to $\rho\prec\rho_0$. 
\end{proof}
\end{theorem}

A physical system coupled to its environment can be described
using a Markovian master equation
in certain parameter regions, most notably in the weak-coupling limit.
In the standard derivation
(see Appendix~\ref{sec:heatbath} and Ch.~3.3 in~\cite{BreuPetr02})
the real part of the Fourier transform of the
bath correlation function enters the Lindbladian dissipator~$\Gamma$,
while the imaginary part induces a Lamb-shift term
$\Hlamb$ to the system Hamiltonian.
The Lamb shift is always switched alongside with the
dissipation~$\Gamma$, with their ratio constant.
While in many systems of interest $\Hlamb$ typically commutes with~$H_0$
(see \cite[Eqn.~(3.142)]{BreuPetr02}), next
we will relax this condition together with the assumption of diagonal drift Hamiltonians $H_0$,
since both were just introduced for simplicity.

\begin{corollary}\label{cor:Trotter-Decoup}
The diagonality requirement for the drift and Lamb-shift
Hamiltonians~$H_0$ and~$\Hlamb$ in the two theorems above
can in fact be relaxed,
since full unitary controllability ensures
one may always decouple
$H_0' := H_0+\Hlamb$
during the evolution under $\Gamma$.
\end{corollary}
\begin{proof} The Trotter limit
$\lim_{k\to\infty} \big(e^{-t(\Gamma + i\hat H_0')/k}
\: e^{-t(-i\hat H_0')/k}\big)^k = e^{-t \Gamma}$
gives an effective evolution under the dissipator~$\Gamma$ alone, and
the compensating terms $e^{t (i\hat H_0')/k}$ can obviously be obtained by unitary control.
\end{proof}

\medskip
Note that the theorems above are stated under further mildly simplifying conditions. So
\begin{enumerate}
\item the theorems  hold {\em a forteriori} if the noise amplitude
	is not only a bang-bang control $\gamma(t)\in\{0,\gamma\}$, but may vary in time within the entire 
	interval $\gamma(t)\in[0,\gamma]$;
\item the theorems are stated  independent of the choice of $\gamma > 0$; however, for the Born-Markov approximation
	to hold such as to give a Lindblad master equation, $\gamma$ may have to be limited according to 
	the considerations in App.~\ref{sec:heatbath} ({\em vide infra}); note that the theorems hold in particular also for $\gamma$
	sufficiently small to ensure adiabaticity~\cite{ABLZ12};
\item if several qubits come with switchable noise of the same type (unital or non-unital),
	then the (closures of the) reachable sets themselves do not alter, yet the control problems can be solved more efficiently;
\item
	a single switchable non-unital noise process equivalent to amplitude damping 
	on top of switchable unital ones suffices to make the system act transitively;
\item for systems with non-unital switchable noise equivalent to amplitude damping,
	the (closure of the) reachable set under non-Markovian conditions 
	cannot grow, since it already encompasses the entire set of density operators 
	(see Sec.~\ref{sec:controllabilities})---yet again the control problems may become easier to solve efficiently; 
\item likewise in the unital case, the reachable set does not grow under non- Markovian conditions, since the
 	Markovian scenario already explores all interconversions
        obeying the majorisation condition (see also Sec.~\ref{sec:controllabilities});
\item the same arguments hold for a {\em coded logical subspace} that is unitarily fully controllable and coupled
	to a single physical qubit undergoing switchable noise.
\end{enumerate}

\section{Generalised Noise Generators, Relation to Finite-Temperature Baths and Algorithmic Cooling}\label{app:B}

In this section we extend our analysis to finite-temperature noise.
We start with a formal mathematical generalization of the Lindblad generators
in Sec.~\ref{sec:genLindblad}, and show how the resulting conditions on reachability
are connected to algorithmic cooling.
Then, in Sec.~\ref{sec:heatbath}, we provide a physical derivation of our example control system
by weakly coupling a chain of qubits to a bosonic (or fermionic) Markovian bath of a finite temperature.
We provide some additional finite-temperature reachability results, and show that our approach
can go beyond algorithmic cooling.
We conclude in Sec.~\ref{sec:Lindblad-w-commuting-Lamb} where we discuss some technicalities
related to the Lamb shift.


\subsection{Generalised Lindblad Generators}\label{sec:genLindblad}
The noise scenarios of the above theorems can be
generalised to the Lindblad generator
$V_\theta:=\left(\begin{smallmatrix} 0 & \cos(\theta/2)\\ \sin(\theta/2) & 0 \end{smallmatrix}\right)$
with $\theta\in[0,\pi/2]$.
The Lindbladian and its exponential are now given by
\begin{equation}\label{eqn:GLtheta}
\Gamma(\theta) = \begin{pmatrix}
  \sin^2(\frac{\theta}{2}) & 0 & 0 & -\cos^2(\frac{\theta}{2})\\[1mm]
  0 &\tfrac{1}{2} & -\tfrac{1}{2}\sin(\theta) & 0\\[1mm]
  0 & -\tfrac{1}{2} \sin(\theta) & \tfrac{1}{2} & 0\\[1mm]
  -\sin^2(\frac{\theta}{2}) & 0 & 0 & \cos^2(\frac{\theta}{2})
\end{pmatrix}
\quad \text{and}
\end{equation}
\begin{equation}\label{eqn:expGLtheta}
e^{-\gamma t \,\Gamma(\theta)} =
\begin{pmatrix}
1 -(1-\varepsilon) \sin^2(\frac{\theta}{2}) &0 &0 & (1-\varepsilon) \cos^2(\frac{\theta}{2}) \\[1mm]
0 &\sqrt{\varepsilon}\cosh(\frac{\gamma t}{2} \sin(\theta)) &\sqrt{\varepsilon}\sinh(\frac{\gamma t}{2} \sin(\theta)) & 0\\[1mm]
0 &\sqrt{\varepsilon}\sinh(\frac{\gamma t}{2} \sin(\theta)) &\sqrt{\varepsilon}\cosh(\frac{\gamma t}{2} \sin(\theta)) & 0\\[1mm]
(1-\varepsilon) \sin^2(\frac{\theta}{2}) &0 &0 & 1 -(1-\varepsilon) \cos^2(\frac{\theta}{2})
\end{pmatrix}
\end{equation}
with $\varepsilon:=e^{-\gamma t}$.
The eigenvalues of $\Gamma(\theta)$ are $\{0, 1\}$
for the outer block
$\left(\begin{smallmatrix} \Gamma_{11} & \Gamma_{14} \\ \Gamma_{41} & \Gamma_{44} \end{smallmatrix}\right)$
pertaining to the evolution of the populations, i.e.\ the diagonal elements of $\rho(t)$, and
$\frac{1}{2}(1 \pm |\sin(\theta)|)$
for the inner block
$\left(\begin{smallmatrix} \Gamma_{22} & \Gamma_{23} \\ \Gamma_{32} & \Gamma_{33} \end{smallmatrix}\right)$
ruling the evolution of the coherences, i.e.\ the off-diagonal elements of~$\rho(t)$.

Choosing the initial state diagonal, the action on a diagonal vector of an $n$-qubit state
takes the following form that can be decomposed into a convex sum of
a pure amplitude-damping part and a pure bit-flip part (cf.~Eqs.~\eqref{eqn:thm1},\eqref{eqn:thm2}):
\begin{equation}
\label{eq:genlindR}
R_\theta(t)
= \unity_2^{\otimes (n-1)} \otimes \left[
\cos(\theta)
\begin{pmatrix}
  1 & (1-\varepsilon)\\
  0 &  \varepsilon
\end{pmatrix}
+\left(1-\cos(\theta)\right)
\frac{1}{2}
\begin{pmatrix}
  (1+\varepsilon) & (1-\varepsilon)\\
  (1-\varepsilon) & (1+\varepsilon)
\end{pmatrix}\right]\;.
\end{equation}

In order to limit the entire dissipative action over some fixed time $\tau$ to the first two eigenvalues (as in Thm.~1), 
one may switch again as in Eqn.~\eqref{eqn:switch-time1} after a time
\begin{equation}\label{eqn:switch-theta}
\tau_{ij}(\theta):=\frac{1}{\gamma}  \;
\ln \left(\frac{e^{\gamma\tau}\left(\frac{\rho_{ii}}{\rho_{jj}} -\tan^2(\frac{\theta}{2})\right)
  +\left(1-\tan^2(\frac{\theta}{2})\frac{\rho_{ii}}{\rho_{jj}}\right) }
    {\left(1-\tan^2(\frac{\theta}{2})\right)\left(\frac{\rho_{ii}}{\rho_{jj}}+1\right)}
\right)
=
\frac{1}{\gamma} \ln \Bigg(
\frac{e^{\gamma\tau}\big(\frac{\rho_{ii}}{\rho_{jj}}\bboltz  -1 \big) +\big(\bboltz -\frac{\rho_{ii}}{\rho_{jj}}\big)} 
{(\bboltz-1)\big(\frac{\rho_{ii}}{\rho_{jj}}+1\big)}
\Bigg)
\quad.
\end{equation}
The result reproduces Eqn.~\eqref{eqn:switch-time1} for $\theta\to 0$
and it coincides with Eqn.~\eqref{eqn:switch-bf} below by identifying
$\tan^2(\frac{\theta}{2}) = \frac{1-\cos(\theta)}{1+\cos(\theta)}$
with a fiducial inverse Boltzmann factor~$\bboltz^{-1} = e^{\beta \hbar \omega}$
describing the effect of the noise on diagonal states.
The switching condition is meaningful as long as
$0 \le \tau_{ij}(\theta) \le \tau$, which corresponds to
the condition
\begin{equation}\label{eqn:theta-stop}
\tan^2(\tfrac{\theta}{2}) \le \frac{\rho_{ii}}{\rho_{jj}} \le \cot^2(\tfrac{\theta}{2})\;.
\end{equation}

\noindent
If $\theta \neq \frac{\pi}{2}$, the noise qubit has the unique fixed-point state
\begin{equation}\label{eqn:fixp-theta}
\rho_\infty(\theta) = 
\begin{pmatrix}
  \cos^2(\frac{\theta}{2}) & 0\\
  0 & \sin^2(\frac{\theta}{2})
\end{pmatrix}\;.
\end{equation}
Note that the parameter $\theta$ corresponds to the inverse temperature
\begin{equation}
\label{eqn:AlgCooling-delta}
\beta(\theta) =
\tfrac{2}{-\hbar \omega} \operatorname{artanh}\left(\delta(\theta)\right)
\quad \text{with} \quad
\delta(\theta) := \cos(\theta) = \frac{b-1}{b+1}\;.
\end{equation}

Thus the relaxation by the single Lindblad generator $V_\theta$ shares the
fixed point with equilibrating the 
system via the noisy qubit with a local bath of temperature
$\beta(\theta)$.
As limiting cases, pure amplitude damping ($\theta=0$) is
brought about by a bath of zero temperature,
while pure bit-flip ($\theta = \tfrac{\pi}{2}$) shares the fixed point with
the infinite-temperature limit.
See Sec.~\ref{sec:heatbath} for the relation to physical heat baths.

In a single-qubit system with unitary control and a bang-bang switchable noise generator~$V_\theta$
it is straightforward to see that
one can (asymptotically) reach all states with purity less or equal to the
larger of the purities of the initial state~$\rho_0$ and~$\rho_\infty(\theta)$
(a special case also treated in~\cite{RBR12}),
\begin{equation}
\overline\reach_{\rm 1 qubit,\Sigma_\theta}(\rho_0) = \{\rho\,|\,\rho\prec\rho_0\} \cup \{\rho'\,|\,\rho'\prec\rho_\infty(\theta)\}\;.
\end{equation}
In contrast, for $n\geq 2$ qubits the situation is more involved:
relaxation of a diagonal state can only be limited to a single pair of
eigenvalues if all the remaining ones can be arranged in pairs each satisfying
Eqn.~\eqref{eqn:theta-stop}.

\subsubsection*{Connection to algorithmic cooling}

However, Eqn.~\eqref{eqn:theta-stop} poses no restriction  in an important special case, i.e.~the task of cooling:
starting from the maximally mixed state, optimal control protocols with period-wise relaxation by $V_\theta$ 
interspersed with unitary permutation of diagonal density operator elements clearly include
the partner-pairing approach~\cite{SMW05} to
\emph{algorithmic cooling} with bias~$\delta(\theta)$ defined in Eqn.~\eqref{eqn:AlgCooling-delta}
as long as $0\leq\theta<\frac{\pi}{2}$. Note that this type of algorithmic cooling proceeds also just on
the diagonal elements of the density operator, but it involves no transfers limited to a single pair
of eigenvalues. Let $\rhoalg$ define the state(s) with highest asymptotic purity achievable by
partner-pairing algorithmic cooling with bias $\delta$.
As the pairing algorithm is just a special case of unitary evolutions plus relaxation brought about by $V_\theta$, 
one arrives at
\begin{equation}
\label{eq:algocool-delta}
\overline\reach(\rho_0) \supseteq \overline\reach(\rhoalg) \quad\text{for any $\rho_0$}\;,
\end{equation}
because any state $\rho_0$ can clearly be made diagonal to evolve into a fixed-point state 
obeying Eqn.~\eqref{eqn:theta-stop}, from whence the purest state $\rhoalg$ can be reached by 
partner-pairing cooling.

To see this in more detail, note that a (diagonal) density operator $\rho_\theta$ of an $n$-qubit system 
is in equilibrium with a bath of inverse temperature $\beta(\theta)$ coupled to its terminal qubit, if the
pairs of consecutive eigenvalues satisfy
\begin{equation}
\frac{\rho_{ii}}{\rho_{i+1,i+1}}=\cot^2(\theta/2)\quad\text{for all odd $i < 2^n$}\;.
\end{equation}
Hence (for $\theta \neq \pi/2$) such a $\rho_\theta$ is indeed a fixed point under uncontrolled drift, i.e.~relaxation by $V_\theta$
and evolution under a diagonal Hamiltonian $H_0$ thus extending Eqn.~\eqref{eqn:fixp-theta} to $n$ qubits. 
Now, if (say) the first pair of eigenvalues is unitarily permuted
(with relaxation switched off),
a subsequent evolution under the drift term only affects the first pair of eigenvalues as
\begin{equation}
\begin{split}
R_\theta(t)
\begin{pmatrix}
  \sin^2(\tfrac{\theta}{2})\\[1mm]
  \cos^2(\tfrac{\theta}{2})
\end{pmatrix}
&=
\left[\begin{matrix}
    1-(1-\varepsilon)\sin^2(\tfrac{\theta}{2}) & (1-\varepsilon)\cos^2(\tfrac{\theta}{2})\\[1mm]
    (1-\varepsilon)\sin^2(\tfrac{\theta}{2}) & 1-(1-\varepsilon)\cos^2(\tfrac{\theta}{2})
\end{matrix}\right]
\begin{pmatrix}
  \sin^2(\tfrac{\theta}{2})\\[1mm]
  \cos^2(\tfrac{\theta}{2})
\end{pmatrix}\\[2mm]
&=
\begin{pmatrix}
  \cos^2(\tfrac{\theta}{2}) +\varepsilon(\sin^2(\tfrac{\theta}{2})-\cos^2(\tfrac{\theta}{2}))\\[1mm]
  \sin^2(\tfrac{\theta}{2}) -\varepsilon(\sin^2(\tfrac{\theta}{2})-\cos^2(\tfrac{\theta}{2}))
\end{pmatrix}
=
\left[\begin{matrix}
    \varepsilon & (1-\varepsilon)\\
    (1-\varepsilon) &   \varepsilon
  \end{matrix}\right] 
\begin{pmatrix}
  \sin^2(\tfrac{\theta}{2})\\[1mm]
  \cos^2(\tfrac{\theta}{2})
\end{pmatrix} \;.
\end{split}
\end {equation}
In other words, the evolution then acts as a $T$-transform on the first eigenvalue pair. 
Since the switching condition Eqn.~\eqref{eqn:theta-stop}
is fulfilled at any time, {\em all} $T$-transformations with  $\varepsilon\in[0,1]$
on the first pair of eigenvalues can be obtained and preserved during transformations
on subsequent eigenvalue pairs.

Hence from any diagonal fixed-point state $\rho_\theta$ (including $\rhoalg$ as a special case),
those other diagonal states (and their unitary orbits) can be reached that arise by pairwise $T$-transforms  
{\em only} as long as the
remaining eigenvalues can be arranged such as to fulfill the stopping condition Eqn.~\eqref{eqn:theta-stop}.
Suffice this to elucidate why for $n\geq 2$ a fully detailed determination of the asymptotic reachable set in the case
of unitary control plus a single switchable $V_\theta$ on one qubit is more  involved.
This will be made more
explicit in Theorem~\ref{th:bath-T-trafo} of the following section.

\subsection{Bosonic (or Fermionic) Heat Baths}\label{sec:heatbath}
Next we consider coupling to Markovian heat baths: while the physically relevant scenario usually pertains to
bosonic baths, fermionic ones can formally be handled analogously. 
Note that realistic baths composed of fermions at low energies typically consist of fermions being excited and 
de-excited rather than created or destroyed. These excitations are (hardcore) bosons and thus fall under the 
Bose case~\cite{Weiss99, HW04}. For this reason, we focus on bosonic baths in what follows.

\subsubsection*{Bath model}
\label{sec:bathmodel}
\noindent
Assume an ohmic bosonic (or fermionic) heat bath with the Hamiltonian
\be
H_{\text{bath}} := \int_0^\infty \wrt{\omega} \omega \, \big(b_\omega^\dagger b_\omega +\tfrac{1}{2}\big)\;,
\ee
in a state representing a canonical ensemble.
The bath coupling operator is
\be
\label{eq:bathcoupling}
B = \int_0^\infty \wrt{\omega} \sqrt{J(\omega)} \big(b_\omega +b_\omega^\dagger\big),
\ee
where
$J(\omega) = \omega f (\frac{\omega}{\omegacut} )$
describes the ohmic spectral density of the bath, and
$f(x) = (1+x^2)^{-1}$ is the Lorentz-Drude cutoff function.
The bath correlation function in the interaction picture generated by $H_{\text{bath}}$ is now
\begin{align}
\notag
\expt{B(s) B(0)}
&=
\int_0^\infty \wrt{\omega} \int_0^\infty \wrt{\omega'}
\sqrt{J(\omega) J(\omega')}
\expt{\big(e^{-i \omega s} b_\omega +e^{i \omega s} b_\omega^\dagger\big) \big(b_{\omega'} +b_{\omega'}^\dagger\big)}\\[2mm]
&=
\int_0^\infty \wrt{\omega} J(\omega)
\big( (1+n(\omega)) e^{-i \omega s} +n(\omega) e^{i \omega s}) \big),
\end{align}
where $n(\omega) = (e^{\beta \hbar \omega} \mp 1)^{-1}$ is the Planck (Fermi) function,
and $\beta=\frac{1}{k_B T_b}$~the inverse temperature of the bath.
Fourier transforming the bath correlation function
and then dividing it into hermitian and skew-hermitian parts,
\be
\label{eq:speccorr}
\speccorr(\omega)
= \int_0^\infty \wrt{s} e^{i s \omega} \expt{B(s) B(0)}
= \frac{1}{2} \gamma(\omega) +i S(\omega),
\ee
one obtains
\be
\gamma(\omega) = \int_{-\infty}^\infty \wrt{s} e^{i s \omega} \expt{B(s) B(0)}.
\ee
Assuming a reasonable bath state that
is both stationary under $H_{\text{bath}}$
and fulfills the Kubo-Martin-Schwinger (KMS) condition, we obtain
\be
\gamma(-\omega) = e^{-\beta \hbar \omega} \gamma(\omega).
\ee
The dissipation rates are given by
\be
\gamma(\omega) = 2 \pi (1 \pm n(\omega)) f(|\omega|/\omegacut)
\begin{cases}
\omega \quad\:\: \text{(bosons)},\\
|\omega| \quad \text{(fermions)}.
\end{cases}
\ee

\subsubsection*{Paradigmatic Ising control system}

Our example control system consists of a chain of qubits with Ising-ZZ coupling,
with each qubit individually resonantly driven.
The $k$th qubit Hamiltonian is (setting $\hbar=1$)
\be
H_k = \omega_k \tfrac{1}{2} Z\sys{k}
+\sum_q \Omega_q(t) \cos(\carrier_q t +\phi_q(t)) X\sys{k}.
\ee
The qubit-qubit nearest-neighbour Ising couplings are given by
\be\label{eqn:H0-Ising}
H_0 = \pi J \sum_{k=1}^{n-1} \tfrac{1}{2} Z\sys{k} Z\sys{k+1}.
\ee
Further assume that
$|\omega_k| \gg |\Omega_q(t)|$,
$|\omega_k| \gg |J|$,
the qubits are driven in resonance
($\carrier_q = \omega_q$), and that
the qubit splittings $\omega_k$ are well separated.
The final, $n$th qubit is further coupled to a heat bath
at inverse temperature~$\beta$, as described in the previous section.
The qubit-bath coupling is
\be
H_{\text{int}} = \kappa(t) A \otimes B,
\ee
where $\kappa(t)$ is a tunable coupling coefficient,
$A = X\sys{n}$,
and $B$~is the bath coupling operator.

Following the standard derivation for the Lindblad equation
using the Born-Markov approximation in the weak-coupling limit~\cite[Ch.~3.3.1]{BreuPetr02},
we transform to the rotating frame generated by
\be
\Hrf = \sum_{k=1}^n \carrier_k \tfrac{1}{2} Z\sys{k} +H_\text{bath}.
\ee
Since the system part of $\Hrf$ is diagonal,
$H_0$ is unaffected by the rotating frame.
Since the qubit splittings $\omega_k$ are well separated one may
apply the RWA with no crosstalk, and the qubit Hamiltonians become
\be
H_k'(t) \approx
\tfrac{1}{2}\Omega_k(t) \big(\cos(\phi_k(t)) X\sys{k} +\sin(\phi_k(t)) Y\sys{k} \big).
\ee

The bath degrees of freedom are traced over, and
will appear in the equation of motion
only as the Fourier transform of the bath correlation function,
see Eq.~\eqref{eq:speccorr}.
Our choice of~$\Hrf$ yields two jump operators
\be\label{eqn:A-omega}
A(\omega_n) = \I \otimes \sigma_{-} \text{\quad and\quad} A(-\omega_n) =  \I \otimes \sigma_{+}.
\ee
One obtains a Lindblad equation for the system in the rotating frame,
with the Hamiltonian
\begin{align}
H_u(t) = H_0 +\tfrac{1}{2} \sum_{k=1}^n \Omega_k(t) \big(\cos(\phi_k(t)) X\sys{k} +\sin(\phi_k(t)) Y\sys{k} \big)
+\Hlamb,
\end{align}
where the Lamb shift is
\be\label{eqn:LS-weak}
\Hlamb = \kappa^2(t) \sum_{\omega} S(\omega) A^\dagger(\omega) A(\omega)
=
\kappa^2(t) \big(S(\omega_n) \ketbra{0}{0}_n +S(-\omega_n) \ketbra{1}{1}_n\big)
\hat{=}
\underbrace{\kappa^2(t) \big(S(\omega_n)-S(-\omega_n)\big)}_{\lambda} \tfrac{1}{2} Z\sys{n},
\ee
and the Lindblad dissipator
\begin{align}
\notag
\Gamma(\rho) &=
-\kappa^2(t) \big(
\gamma(\omega_n)
\underbrace{
\Big(
(\I \otimes \sigma_{-}) \rho (\I \otimes \sigma_{-}^\dagger)
-\frac{1}{2} \{\ketbra{0}{0}_n, \rho \}
\Big)
}_{-\Gamma_{\sigma_-}(\rho)}
+\gamma(-\omega_n)
\underbrace{
\Big(
(\I \otimes \sigma_{+}) \rho (\I \otimes \sigma_{+}^\dagger)
-\frac{1}{2} \{\ketbra{1}{1}_n, \rho \}
}_{-\Gamma_{\sigma_+}(\rho)}
\Big)\\
&=
\underbrace{\kappa^2(t) \gamma(\omega_n) (b+1)}_{\gamma}
\;\,
\underbrace{
\Big(
\frac{1}{\bboltz+1} \Gamma_{\sigma_-}
+\frac{1}{\bboltz^{-1}+1} \Gamma_{\sigma_+}
\Big)}_{\Gamma'} (\rho),
\end{align}
where we have introduced
$\bboltz := e^{-\beta \hbar \omega_n}$
to denote the Boltzmann factor of the bath-coupled qubit.
Note that regardless of the bath coupling coefficient $\kappa(t)$,
the ratio of the Lamb shift magnitude $\lambda$ and the dissipation rate $\gamma$ is given by
\be
\frac{\lambda}{\gamma} =
\frac{S(\omega_n)-S(-\omega_n)}{2 \gamma(\omega_n) (b+1)}.
\ee
This ratio only depends on $\bboltz$ and $|\omega_n|/\omegacut$.
The latter we in our finite-temperature examples fix to~$1/5$ (somewhat arbitrarily).
Note that in our convention the ground state of a spin-$\tfrac{1}{2}$-qubit
is $\ket{0} := (\begin{smallmatrix} 1\\0\end{smallmatrix})=\ket{\negthinspace\uparrow}$ in agreement with~\cite{EBW87}.
Thus we must have $\omega_n < 0$,
and for nonnegative temperatures $\bboltz \ge 1$.

Acting on the Liouville space of the final qubit, the Lindblad superoperator $\Gamma'$ takes the form
\be\label{eqn:bath-liouv}
\Gamma'
=
\frac{1}{\bboltz+1}
\begin{pmatrix}
1 & 0 & 0 & 0\\
0 & \frac{1}{2} & 0 & 0\\
0 & 0 & \frac{1}{2} & 0\\
-1 & 0 & 0 & 0
\end{pmatrix}
+
\frac{1}{\bboltz^{-1}+1}
\begin{pmatrix}
0 & 0 & 0 & -1\\
0 & \frac{1}{2} & 0 & 0\\
0 & 0 & \frac{1}{2} & 0\\
0 & 0 & 0 & 1
\end{pmatrix}
=
\begin{pmatrix}
\frac{1}{\bboltz+1} & 0 & 0 & \frac{-1}{\bboltz^{-1}+1}\\
0 & \frac{1}{2} & 0 & 0\\
0 & 0 & \frac{1}{2} & 0\\
\frac{-1}{\bboltz+1} & 0 & 0 & \frac{1}{\bboltz^{-1}+1}
\end{pmatrix}.
\ee
The eigenvalues of $\Gamma'$ are again $\{0, 1\}$
for the outer block and $\{\frac{1}{2}, \frac{1}{2}\}$
for the inner block.
In the zero-temperature limit $T_b\to 0+$ one finds $\bboltz \to \infty$
and thus only the $\Gamma_{\sigma_+}$ term remains.
This corresponds to the purely amplitude-damping case, with
$\ketbra{0}{0}$ in the kernel of the Lindbladian.
In the limit $T_b \to \infty$ we have
$\bboltz \to 1$, and hence one obtains
\be
\lim_{T_b \to \infty} \Gamma'
=: \Gamma_\infty
=
\tfrac{1}{2}
\begin{pmatrix}
1 & 0 & 0 & -1\\
0 & 1 & 0 & 0\\
0 & 0 & 1 & 0\\
-1 & 0 & 0 & 1
\end{pmatrix}
=
\tfrac{1}{2}\;
\Gamma_{\{\sigma_+, \sigma_-\}}.
\ee
This is equivalent to dissipation under the two Lindblad generators
$\{X, Y\}$ since
$\Gamma_{\{X, Y\}} = 2 \Gamma_{\{\sigma_+, \sigma_-\}}$.
In contrast, $X$ as the only Lindblad generator (generating bit-flip noise) gives
\be
\Gamma_{\{X\}}
=
\begin{pmatrix}
1 & 0 & 0 & -1\\
0 & 1 & -1 & 0\\
0 & -1 & 1 & 0\\
-1 & 0 & 0 & 1
\end{pmatrix}
\ee
in agreement with Eqn.~\eqref{eqn:GLtheta}.
The propagators $e^{-\gamma t \Gamma_\nu}$ generated by $\Gamma_{\{\sigma_+, \sigma_-\}}$
{\em vs.} $\Gamma_{\{X\}}$
act indistinguishably
on {\em diagonal} density operators (and those with purely
imaginary coherence terms $\rho_{12}=\rho^*_{21}$). However,
the relaxation of the real parts of the coherence terms $\rho_{12}=\rho^*_{21}$ differs: 
$\Gamma_{\{X\}}$ leaves them invariant
while $\Gamma_{\{\sigma_+, \sigma_-\}}$ does not.
In other words, $\Gamma_{\{X\}}$ has the nontrivial
invariant subspaces used in Eqn.~\eqref{eqn:protect},  while $\Gamma_\infty$  does not.


The evolution of a diagonal state remains diagonal
under the general dissipator of Eqn.~\eqref{eqn:bath-liouv}.
The restriction of~$\Gamma'$ to the diagonal subspace is given by
\be
\Gamma''
=
\begin{pmatrix}
\frac{1}{\bboltz+1} &\frac{-1}{\bboltz^{-1}+1}\\[2mm]
\frac{-1}{\bboltz+1} & \frac{1}{\bboltz^{-1}+1}
\end{pmatrix}.
\ee
It has the eigenvalues $\{0,1\}$, which makes it idempotent, and thus the
corresponding propagator (using $\varepsilon(t) := e^{-\gamma t}$)
is given by
\be
\label{eqn:bf-evol}
R_T(t) := e^{-\gamma t \Gamma''} = \I +(\varepsilon(t)-1) \Gamma''\;.
\ee

The noise propagator in the diagonal subspace is again a convex sum of
a pure amplitude-damping part and a pure bit-flip part
(cf.~Eqs.~\eqref{eqn:thm1},\eqref{eqn:thm2},\eqref{eq:genlindR}):
\be
R_T(t) =
\unity_2^{\otimes (n-1)} \otimes \left[
\frac{b-1}{b+1}
\begin{pmatrix}
  1 & (1-\varepsilon)\\
  0 &  \varepsilon
\end{pmatrix}
+\frac{2}{b+1}\cdot
\frac{1}{2}
\begin{pmatrix}
  (1+\varepsilon) & (1-\varepsilon) \\
  (1-\varepsilon) & (1+\varepsilon)
\end{pmatrix}
\right].
\ee

Now we generalise the switching condition of  Eqn.~\eqref{eqn:switch-time1} for a
finite-temperature bosonic (fermionic) bath:
we ask under which conditions one can undo the dissipative evolution for the populations of level
$i$ and $j$ over a fixed time $\tau$
by letting the system evolve for a time $\tau_{ij}$,
then swapping the populations of levels $i$ and $j$, letting the system evolve for the remaining time $(\tau-\tau_{ij})$
and swapping again. Denoting the population swap by~$Q$, for any
$\rho_{ii}, \rho_{jj}$ and $\tau$, in a formal step one thus has to find $\tau_{ij}$ such that
\be
Q \circ R_T(\tau-\tau_{ij}) \circ Q \circ R_T(\tau_{ij})
\begin{pmatrix}
\rho_{ii}\\
\rho_{jj}
\end{pmatrix}
= 
\begin{pmatrix}
\rho_{ii}\\
\rho_{jj}
\end{pmatrix}.
\ee
A somewhat lengthy but straightforward calculation yields the {\em switching-time condition}
for $\tau_{ij}$ reading
\be
\label{eqn:switch-bf}
\tau_{ij}
=
\frac{1}{\gamma} \ln \Bigg(
\frac{e^{\gamma\tau}\big(\frac{\rho_{ii}}{\rho_{jj}}\bboltz -1\big) +\big(\bboltz -\frac{\rho_{ii}}{\rho_{jj}}\big)} 
{(\bboltz-1)\big(\frac{\rho_{ii}}{\rho_{jj}}+1\big)}
\Bigg),
\ee
which in the limit of $T \to 0+$ (i.e.\ $\bboltz \to \infty$) reproduces Eqn.~\eqref{eqn:switch-time1} and
in its general form it coincides with Eqn.~\eqref{eqn:switch-theta}.
Limiting the switching time $\tau_{ij}$ to the physically meaningful interval
$0\leq\tau_{ij}\leq\tau$ then translates into a {\em stopping-time condition} for
$\tau$ related to the admissible population ratios reading
\begin{equation}
\label{eqn:bf-stab}
\bboltz^{-1}
\leq \frac{\rho_{ii}}{\rho_{jj}} \leq
\bboltz\;,
\end{equation}
which turns into a limitation only for non-zero temperatures (i.e.\ finite Boltzmann factors~$b$).
Hence for both bosonic and fermionic baths the relaxative transfer
between $\rho_{ii}$ and $\rho_{jj}$ can only be undone if their ratio
falls into the above interval.
On the other hand, in terms of a single qubit coupled to the bosonic (fermionic) bath
this population ratio would relate to
the thermal equilibrium state
\begin{equation}
\rho_{\text{eq}}
:=
\frac{1}{1+\bboltz}\begin{pmatrix} \bboltz & 0\\ 0 & 1\end{pmatrix}\quad.
\end{equation}

Another way of understanding~Eq.~\eqref{eqn:bf-stab} is in terms of temperatures.
The population ratio $\rho_{ii}/\rho_{jj}$ corresponds to an effective temperature
$T_q = \frac{-\hbar \omega_n}{k \ln(\rho_{ii}/\rho_{jj})}$
(possibly negative).
$T_q$ will always thermalize towards~$T_b$
(if negative, $T_q$ will fall to $-\infty$, wrap around to $+\infty$ and then fall towards~$T_b$),
as illustrated in Fig.~\ref{fig:thermalization}.
A population swap negates~$T_q$.
Thus, iff $|T_q| < T_b$, the thermalization cannot be reversed (this
is a no-return zone, bath is more ``entropic'' than the state).
Also, $T_q$ will never enter the forbidden zone $|T| < \min(|T_q|, T_b)$.

To sum up, in baths of finite~$T_b$ one can always
reverse dissipative population transfer between passive population pairs
(as in Theorem~\ref{thm:transitivity}) for some values of~$T_q$, 
but one cannot protect the inactive population pairs (as in Theorem~\ref{thm:majorisation}),
because the noise lacks the protected subspaces pure bit-flip noise has.

\begin{figure}[!ht]
\includegraphics[width=0.6\columnwidth]{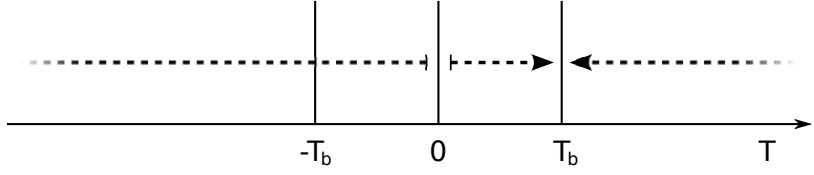}
\caption{\label{fig:thermalization}
Thermalization of a qubit coupled to a bath of temperature~$T_b$.
}
\end{figure}

\subsubsection*{Further reachability results}

\begin{theorem}
\label{th:bath-T-trafo}
Consider a system of $n$~qubits with a drift Hamiltonian~$H_0$ leading to a connected coupling topology,
where the local controls and~$H_0$ suffice to give full unitary control.
If one of the $n$ qubits is in a switchable way  coupled to a  bosonic or fermionic bath 
(of any temperature)
such that the Lamb-shift term $\Hlamb$ commutes with the system Hamiltonian $H_0$,
one finds
\begin{enumerate}
\item[(1)] In the single-qubit case ($n=1$), any $T$-transform can be performed; in particular,
	any pair $(\rho_{ii},\rho_{jj})$ can be averaged (in a perfect way in finite time if the temperature is finite).
\item[(2)] For several qubits ($n\geq 2$), any $T$-transform can be performed only as long as
	all the other level populations can be arranged in pairs fulfilling the stopping condition  Eqn.~\eqref{eqn:bf-stab}.
\end{enumerate}
\end{theorem}

\begin{proof}
Unitary permutations can always be used to arrange the population pair in a suitable location.

\begin{enumerate}
\item[(1)]
In a single qubit, one may formally identify Eqn.~\eqref{eqn:bf-evol} with a $T$-transform,
\begin{equation}
\left(
\I +(\varepsilon(\tau(\lambda)) -1) \Gamma''
\right)
\begin{pmatrix} \rho_{ii} \\ \rho_{jj} \end{pmatrix} =
\left((1-\lambda) \I +\lambda Q\right)
\begin{pmatrix} \rho_{ii} \\ \rho_{jj} \end{pmatrix}
\end{equation}
and solve for $\varepsilon$, yielding
\begin{equation}\label{eqn:bf-eps}
\varepsilon(\tau(\lambda))
= e^{-\gamma \tau(\lambda)}
=
1 -\lambda (1+\bboltz)
\frac{1-(\rho_{ii}/\rho_{jj})}
{\bboltz -(\rho_{ii}/\rho_{jj}) }
\,.
\end{equation}
Assume first that~\eqref{eqn:bf-stab} is fulfilled.
We may ensure that $\rho_{ii}/\rho_{jj} \le 1$ by doing a unitary
\SWAP{} if necessary, and
can always find a physical $0 < \varepsilon \le 1$ for any
$0 \le \lambda \le 1$.
If \eqref{eqn:bf-stab} is not fulfilled,
the $0 \le \lambda \le 1$ parameter interval can be divided into two subintervals, one
requiring a \SWAP{}, the other one not.
The dividing $\lambda$ point between the subintervals (corresponding to
$\rho_{\text{eq}}$)
can only be reached in the limit~$\gamma \tau \to \infty$.

\item[(2)] For $n$-qubit systems (with $n\geq 2$), the second assertion then follows 
	provided all passive population pairs can be ordered such
	that they simultaneously fulfill the stopping conditon Eqn.~\eqref{eqn:bf-stab}.
\end{enumerate}
\end{proof}

Generalising Example 1 from zero temperature to finite temperatures (in some analogy to algorithmic cooling~\cite{SMW05}), 
one can cool the maximally mixed state $\tfrac{1}{N}\unity_N$ to approximate the ground state by the following diagonal state
\begin{equation}\label{eqn:cool-bf}
\rhoalg := \frac{1}{Z} \diag\left(
b^{\frac{N}{2}}, b^{\frac{N}{2}-1}, b^{\frac{N}{2}-1}, b^{\frac{N}{2}-2}, b^{\frac{N}{2}-2}, \dots, b^k, b^k, \dots, b^1, b^1, b^0
\right)\,,
\end{equation}
where the partition function takes the form
$Z := 1 +2\frac{b -b^{N/2}}{1 -b} +b^{N/2}$.
One finds the following conservative inclusions for the reachable set at {\em finite temperatures} $0<T<\infty$
(while the limiting case $T\to 0$ is exactly settled by Theorem~1 and
the bit-flip analogue to $T \to \infty$ yet with a single Lindblad generator is settled by Theorem~2
\footnote{NB: Theorem 2 with a single Lindblad generator does not correspond to a bath with
$T \to \infty$ and two Lindblad generators, since the latter lacks the protected subspaces.
See the paragraph prior to Theorem~\ref{th:bath-T-trafo}.}).
\begin{theorem}\label{thm:reach-with-bath}
For both bosonic and fermionic baths of finite temperatures $0<T<\infty$ and any initial state $\rho_0\in\pos_1$, the following observations hold in $n$-qubit systems 
(otherwise satisfying the conditions of Theorems 1 and 2):
\begin{enumerate}
\item[(1)] From any initial state $\rho_0\in\pos_1$, the maximally mixed state $\tfrac{1}{N}\unity$ can be 
	reached by averaging.
\item[(2)] Regardless of the initial state, from the maximally mixed state $\tfrac{1}{N}\unity$ in turn, 
	at least a state of the purity of the one by algorithmic cooling, $\rho_{\rm alg}$ of Eqn.~\eqref{eqn:cool-bf}, 
	can be reached~\footnote{Test~1 below provides numerical evidence that one can indeed go beyond algorithmic cooling.
}.
\item[(3)] From $\rho_{\rm alg}$ (or the purest diagonal state $\rho_{\rm p}$ reachable) all those states
	$\rho\prec \rho_{\rm alg}$ (or $\rho\prec \rho_{\rm p}$) can be reached that can be obtained by a 
	sequence of $T$-transforms with each step fulfilling the stopping condition of Eqn.~\eqref{eqn:bf-stab}.
\end{enumerate}
\end{theorem}
\begin{proof}
(1) From any initial state $\rho_0\in\pos_1$, the maximally mixed state $\tfrac{1}{N}\unity$ can be reached by
averaging, since by
Theorem~\ref{th:bath-T-trafo} we can always average any single pair of eigenvalues $(\rho_{ii},\rho_{jj})$.
All the eigenvalue pairs of $\rho_0$ can be averaged noting that
the pairs reaching the average can each be stabilised according to Eqn.~\eqref{eqn:bf-stab}.
After $n$ rounds of
averaging and sorting the max.\ mixed state is obtained, compare also the erasing task in Example~2 
of Appendix~\ref{sec:furtherresults} below.
\\[1mm]
(2) From $\tfrac{1}{N}\unity$ in turn, the diagonal state $\rho_{\rm alg}$ of Eqn.~\eqref{eqn:cool-bf} can be reached: 
In analogy to algorithmic cooling \cite{SMW05}
(generalising Example~1 of the main text) after the first round of equilibration,
half of the populations are (up to normalisation) proportional to $\bboltz$, the other half to~$1$.
This procedure of sorting and splitting degenerate eigenvalues is to be repeated more than $2^n$ times in $n$-qubit systems. 
Note that after sorting by descending magnitude, neighbouring eigenvalues always obey 
$b^{-1} \leq \tfrac{\rho_{ii}}{\rho_{i+1,i+1}}\leq b$.
Yet finally the inner pairs of eigenvalues end up degenerate as in Eqn.~\eqref{eqn:cool-bf}, because after
sorting by descending magnitude, pairs $(\rho_{ii},\rho_{i+3,i+3})$ cannot be stabilised, because for these pairs
the switching condition Eqn.~\eqref{eqn:switch-bf} is no longer fulfilled.
\\[1mm]
(3) Direct consequence of point (2) in Theorem~\ref{th:bath-T-trafo} above.
\end{proof}

It is important to note that Theorem~\ref{thm:reach-with-bath} gives but a {\em conservative}
estimate of the actual reachable set. It is limited by what we can readily prove rather than
what one can physically achieve: numerical evidence shows by {\bf Test~1} below that states purer than
the ones by algorithmic cooling can be obtained, and by {\bf Test~2} it also shows
that one may go beyond $T$-transforms---a probable common reason being that state transfer is not
limited to happen between diagonal states only.

\medskip
One may compare the performance of a switchable coupling between a heat bath and
(a single qubit of) the system
with algorithmic cooling by studying certain illuminating test cases numerically.
Our test system is the one described in the beginning of this section,
a linear chain of qubits with Ising-type nearest-neighbor couplings with local $x$ and~$y$ controls on each qubit.
The system is coupled to a bosonic heat bath.
The bath coupling is assumed tunable with $\gamma(t)\in[0, 5 J]$.
In this setting, we pursue the following two types of test questions:
\begin{enumerate}
\item[{\bf Test 1:}]
Starting from $\rho = \unity_N/N$, can one get closer to $\ketbra{0}{0}$, i.e.\ to a purer state,
using switchable noise than by using algorithmic cooling, which ends up in the state $\rhoalg$?

\item[{\bf Test 2:}]
Starting from $\rhoalg$, how close can one get to the state
$\rho_x = \frac{1}{Z} \diag\left(1, x, x, \ldots\right)$ by virtue of switchable noise, where we define
$x = (Z-1) / (N-1)$?
Here $\rho_x \prec \rhoalg$, and $\rho_x$ is the state obtained by averaging
all the eigenvalues of $\rhoalg$ except the highest one.
\end{enumerate}

\begin{table}[h!]
\caption{Noise-Switching Going beyond Algorithmic Cooling: Numerical Results}\vspace{2mm}
\begin{tabular}{l c c c c}
\hline\hline\\[-5mm]
 && \multicolumn{2}{c}{\rule[.1mm]{60mm}{1pt}} & \rule[.1mm]{60mm}{1pt} \\[1mm]
 && \multicolumn{2}{c}{\bf Test~1$^*$} & {\bf Test~2}\\[-2mm]
System & Boltzmann & \multicolumn{2}{c}{\rule[.1mm]{60mm}{1pt}} & \rule[.1mm]{60mm}{1pt} \\[1mm]
$\sharp$ qubits\phantom{XX} & factor $\bboltz=e^{\beta \hbar |\omega_n|}$ && Ground state population & Frobenius norm error $\delta_F$ \\[1mm]
\hline\\[-3mm]
1~qubit & $4$
& no &(going beyond $\rhoalg$ seems impossible)
& $0$\\[1mm]
\hline\\[-3mm]
2~qubits
& $2$
& yes & $0.45336 > \frac{4}{9} \approx 0.44444$ 
& $4.0604 \times 10^{-3}$\\
& $4$
& yes & $0.65308 > \frac{16}{25} \approx 0.64000$
& $6.5486 \times 10^{-3}$\\
& $8$
& yes & $0.80004 > \frac{64}{81} \approx 0.79012$
& $3.0545 \times 10^{-3}$\\[2mm]
\hline\\[-3mm]
3~qubits
& $2$
& yes & $0.35575 > \frac{2^4}{3^2\cdot5} \approx 0.35556$
& -\\[0.9mm] 
& $4$
& yes & $0.60833 > \frac{2^8}{5^2\cdot17} \approx 0.60235$
& -\\[0.9mm] 
& $8$
& yes & $0.78009 > \frac{2^{12}}{3^4 \cdot 5 \cdot 13} \approx 0.77797$
& -\\[0.9mm] 
& $16$
& yes & $0.88503 > \frac{2^{16}}{17^2 \cdot 257} \approx 0.88237$
& -\\[0.9mm]
& $32$
& yes & $0.94139 > \frac{2^{20}}{3^2 \cdot 5^2 \cdot 11^2 \cdot 41} \approx 0.93939$
& -\\[0.9mm]
\hline\hline\\[-5mm]
\multicolumn{5}{c}{{\scriptsize{*) NB: Test~1 compares {\em finite-time} optimal control to
algorithmic cooling with cooling intervals of {\em infinite length} for perfect exponential decay.}}
}\\
\end{tabular}
\end{table}


When comparing unitary control extended with switchable noise to algorithmic cooling,
it is important to note that our control schemes come with two advantages: (i) not only
can one achieve (slightly) higher purities than by algorithmic cooling, but even more
importantly (ii) cooling and unitary control may proceed {\em simultaneously}. One thus arrives
at shorter and more efficient control sequences as compared to those manual
(or paper-and-pen) approaches, where cooling and unitary population transfer have to be separated.

\subsection{Further Remarks on Lamb Shifts and Time Dependence of the Lindbladian}
\label{sec:Lindblad-w-commuting-Lamb}

For completeness, here we recollect results for two standard cases:
(a) the usual weak-coupling limit~\cite{BreuPetr02} to baths
covering the entire temperature range and 
(b) the singular coupling limit~\cite{GK76,FG76,BreuPetr02} which concomitantly invokes the 
high-temperature limit~\cite{FG76,GFKVS78}.
We finally comment on the distinction to the adiabatic scenario in~\cite{ABLZ12}.

(a)
In the standard approach to the weak-coupling limit case,
invoking the rotating-wave approximation (RWA)
results in a Lamb shift that commutes with the full system Hamiltonian~$H_S$,
as the jump operators $A(\omega)$
are by construction eigenoperators of~$H_S$.
Consequently, in a time-dependently controlled system both
the dissipator and the Lamb shift often depend on~$H_u(t)$ and thus
are time-dependent as well, typically in a nonlinear way, see for instance
Eqs.~(50) and~(55) in~\cite{ABLZ12}.
In the derivation of Eqn.~\eqref{eqn:LS-weak} above we used a rotating frame
generated by the local $Z$~terms alone to obtain simple, time-independent
jump operators (Eqn.~\eqref{eqn:A-omega}) of the required form.
This approximation is justified since we assume that both the local controls
as well as the $J$~coupling are weak compared to the qubit splittings.
The resulting Lamb shift,
\be
\Hlamb =
\kappa^2(t) \big(S(\omega_n)-S(-\omega_n)\big) \tfrac{1}{2} Z\sys{n}\,,
\ee
still commutes with the Ising-type drift Hamiltonian~$H_0$ of Eqn.~\eqref{eqn:H0-Ising}.

(b)
In the singular coupling limit~\cite{Davies74,GK76,FG76},
we also obtain a master equation of the Lindblad form~\cite[Ch.~3.3.3]{BreuPetr02}.
With a system-bath coupling
$H_\text{int} = \kappa(t) A \otimes B$,
we obtain a Lamb shift
\be
\Hlamb = 
		\kappa^2(t)\; S(0)\; A^2
\ee
and a Lindblad dissipator (take $V=A$)
\be
\Gamma(\rho)
= -\kappa^2(t) \;\gamma(0)\; \Big(A \rho A -\frac{1}{2}\{A^2, \rho\} \Big)
\ee
where for $S(0)$ and $\gamma(0)$
the frequency argument is kept for completeness, yet it is redundant once assuming singular
coupling i.e.\ a delta-correlated bath. This is because the singular coupling limit
(the physical realization of which requires infinite temperature) comes with strictly white noise
entailing $S$ and $\gamma$ are frequency-independent.
Unlike in the weak-coupling case, 
where $\Hlamb$ always commutes with the system Hamiltonian~$H_S$ 
(compare \cite[Eqn.~(3.142)]{BreuPetr02}),
in the singular coupling case the commutativity of the Lamb shift depends on the system:
In our model system consisting of a string of qubits as described in Sec.~\ref{sec:heatbath},
we have $A^2 = \I$ and consequently $\Hlamb$ effectively vanishes.

In contrast, in the adiabatic regime~\cite{ABLZ12} one has slow large-amplitude sweeps
which thus go beyond an RWA with respect to {\em a single constant carrier frequency}. 
This scenario does not relate to our work.

\section{Suggested Implementation by GMons with Switchable Coupling to Open Transmission Line}
\label{App:GMon}

Superconducting qubits have gone through various iterations of designs,
starting from intuitive ones with macroscopically distinct basis states
like charge and flux qubits, to robust designs like the
transmon~\cite{Koch_2007, Schreier_2008}.
Transmons are weakly anharmonic oscillators whose energy spectrum is insensitive
to slow charge noise---not only are they operated at a flat operating
point of the parametric spectrum, but the total bandwidth of charge
modulation is exponentially suppressed, so even higher-order noise contributions
are small. These advances have led to superior coherence properties~\cite{Paik_2011, DRAG}.
In practical implementations, one has to be aware that the non-computational energy
levels are relatively close by, which we consider to be a non-fundamental
practicality for now.

In the first generation of proposals for superconducting qubits, it
was highlighted that {\em in situ} tunable couplers would be a desirable
feature of the new platform. Still, for the next few generations of
experiments, fixed coupling that could be made effective or noneffective
by tuning the relative frequencies of the qubits was implemented.
This technique is also used in the popular circuit QED approach, where
the coupling element is a spatially distributed resonator and qubits
are frequency-tuned relative to it. Tunable couplers were implemented
for flux qubits~\cite{Mooji_1999,Orlando_1999,Makhlin_1999,Makhlin_2001,Plourde_2004,Hime_2006}.

The fast tunable-coupler-qubit design devised in the Martinis group is 
called GMon~\cite{Mart09, Mart13,Mart14}. 
They have implemented tunable couplings with rather similar parameters between qubits and
between transmission lines (one of them open)~\cite{Mart14}, and there is
no reason why the same should not work between a line and a qubit. 
Recently, the GMon has solved a lot of technological challenges,
rendering it an effective tunable-coupling strategy between qubits and resonators, 
which has also been achieved in~\cite{Hoffman_2011, Srinivasan_2011}. ---
In order to be more realistic, we will treat the GMons as effective qutrit systems:
note that here $H_0$ does not commute with $\Hlamb$
nor is $H_0$ diagonal, which poses no problem of principle, as discussed in App.~A, p3, under point (1).

\subsubsection*{GMon control system}

Here we assume a simple control scheme, where the chain of GMons is
driven by a single microwave signal, and individual GMons are tuned in and
out of resonance by adjusting their individual level splittings.
GMons can be approximated as three-level systems with the Hamiltonian (setting $\hbar=1$)
\be
H_k := \omega_k(t) a_k^\dagger a_k -\Delta_k \ketbra{2}{2}_k
+\Omega(t) \cos(\carrier t +\phi(t)) (a_k+a_k^\dagger),
\ee
where the level splittings $\omega_k(t)/(2\pi)$ are tunable in the range $3$--$10$~GHz,
the anharmonicity $\Delta_k/(2\pi) \approx 400$~MHz,
$\Omega(t)$ is the amplitude and $\phi(t)$ the phase of a microwave drive
at carrier frequency~$\carrier$,
and $a$~is the truncated lowering operator
\be
a := \begin{pmatrix}
0 & 1 & 0\\
0 & 0 & \sqrt{2}\\
0 & 0 & 0
\end{pmatrix}.
\ee
Our system consists of $n$~GMons in a line, with nearest-neighbor couplings given by
\be
H_0 := \pi J \sum_{k=1}^{n-1} \frac{1}{2} \big(a\sys{k}^\dagger a\sys{k+1}
+a\sys{k}a\sys{k+1}^\dagger \big),
\ee
where $J \approx 4 \times 40$~MHz. 
The final GMon is further coupled to an open transmission line which
functions as an ohmic bosonic bath at inverse temperature~$\beta$,
as described in Sec.~\ref{sec:bathmodel}.
The GMon-line coupling is
\be\label{eqn:H-int}
H_{\text{int}} := \kappa(t)\; A \otimes B
\ee
where $\kappa(t)$ is a tunable coupling coefficient,
$A = a\sys{n} +a\sys{n}^\dagger$, and
$B$ is the bath coupling operator.
The bath spectral density cutoff frequency $\omegacut/(2\pi) \approx 40$~GHz.
Concerning tunability,
the physical parameters providing the qubit-bath interaction are set
by the electrical parameters of the fabricated system such as the line
impedance and the coupling capacitance, which do not change under the
external controls. There can be a weak dependence of the inverse
capacitance and inductance matrices entering the Hamiltonian during
the time-dependent tuning of the qubits which only influences the
coupling to the line in the regime of ultra-strong coupling between
the elements, which is anyway incompatible with other assumptions in
this paper.

Another comment may be in order here: In a real environment,
relaxation occurs into the homogeneous half-open transmission line
terminated by a matched resistor, i.e., a resistor connected to the
line without reflections. This is a paradigmatic realization of an
ohmic heat bath. The information contained in the dissipated photons is
not used and hence is not straightforwardly observed.
In the low-temperature
case, qubits decay through spontaneous emission hence creating a
photon of exponential spatial-temporal shape that is absorbed by the
resistor. Detecting these photons would only be possible by an
elaborated setup from open-line circuit QED.

With these stipulations---and justifying the Markov-approximation in the next paragraph---we may 
derive the Lindblad equation by using the Born-Markov approximation in the weak-coupling limit
after transforming to the rotating frame generated by
\be
\Hrf = \sum_{k=1}^n \carrier\; a\sys{k}^\dagger a\sys{k} +H_\text{bath}.
\ee
As in Eqn.~\eqref{eq:speccorr}, after tracing over the bath we are left with
the Fourier transform $\speccorr(\omega)$ of the bath correlation function,
which is separated into its hermitian part $\gamma(\omega)$ and skew-hermitian part $S(\omega)$.

Since $[a^\dagger a, a] = -a$, we have
$e^{i\Hrf t} a\sys{k} e^{-i\Hrf t} = e^{-i \carrier t} a\sys{k}$.
Thus $H_0$ is unaffected by the rotating frame.
Our choice of~$\Hrf$ yields two jump operators,
$A(\carrier) = a\sys{n}$ and $A(-\carrier) = a\sys{n}^\dagger$.
Assuming that $\carrier$ is large enough for the RWA to hold, we
obtain the Lindblad equation for the system in the rotating frame,
with the Hamiltonian
\begin{align}
H_u(t) &=
H_0
+\sum_{k=1}^n \Big[(\omega_k(t)-\carrier) a\sys{k}^\dagger a\sys{k} -\Delta_k \ketbra{2}{2}_k
+\tfrac{1}{2} \Omega(t) \big(e^{i \phi(t)} a\sys{k} + e^{-i \phi(t)} a\sys{k}^\dagger\big)\Big]
+\Hlamb\\[2mm]
&=
H_0
+\sum_{k=1}^n \Big[(\omega_k(t)-\carrier) a\sys{k}^\dagger a\sys{k} -\Delta_k \ketbra{2}{2}_k
+\tfrac{1}{2} \Omega(t) \big(\cos(\phi(t)) (a\sys{k}+a\sys{k}^\dagger) +\sin(\phi(t))
i  (a\sys{k} -a\sys{k}^\dagger \big)\Big]
+\Hlamb
\end{align}
where the Lamb shift is
\be
\Hlamb =
\kappa^2(t) \; \Big(S(\carrier) a\sys{n}^\dagger a\sys{n} +S(-\carrier) a\sys{n} a\sys{n}^\dagger\Big)
\hat{=}
\underbrace{\kappa^2(t) \; \Big(S(\carrier) +S(-\carrier) \Big)}_{\lambda} a\sys{n}^\dagger a\sys{n},
\ee
and the Lindblad dissipator
\begin{align}
\label{eq:gmon-diss}
\Gamma(\rho) &=
-\kappa^2(t) \Big[
\gamma(\carrier)
\underbrace{
\left(
a\sys{n} \rho a\sys{n}^\dagger
-\tfrac{1}{2} \{a\sys{n}^\dagger a\sys{n}, \rho \}
\right)
}_{-\Gamma_{a}(\rho)}
+\gamma(-\carrier)
\underbrace{
\left(
a\sys{n}^\dagger \rho a\sys{n}
-\tfrac{1}{2} \{a\sys{n} a\sys{n}^\dagger, \rho \}
\right)
}_{-\Gamma_{a^\dagger}(\rho)}
\Big]\\[2mm]
&=
\underbrace{2 \kappa^2(t) \gamma(\carrier) (b+1)}_{\gamma}
\;\,
\underbrace{
\tfrac{1}{2}
\Big(
\frac{1}{\bboltz+1}
\Gamma_{a}
+\frac{1}{\bboltz^{-1}+1}
\Gamma_{a^\dagger}
\Big)}_{\Gamma'}(\rho)
\end{align}
with
the Boltzmann factor $\bboltz = e^{-\beta \hbar \carrier} < 1$.


We obtain for the ratio of the Lamb-shift magnitude $\lambda$ and the dissipation rate $\gamma$
\be
\frac{\lambda}{\gamma} =
\frac{S(\carrier)+S(-\carrier)}{2 \gamma(\carrier) (b+1)}
= -\frac{1}{4} \cdot \frac{1-b}{1+b} \cdot \frac{\omegacut}{\carrier}.
\ee

For a single-GMon system ($n=1$), in the absence of driving ($\Omega(t) = 0$)
the Hamiltonian~$H(t)$ is diagonal, and we obtain the
instantaneous decay rates (a.k.a. Einstein coefficients)
\begin{align}
\Gamma_{1\to 0} &= \bra{0}\,\Gamma(\ketbra{1}{1})\,\ket{0} =
\kappa^2 \gamma(\carrier)\\
\Gamma_{0\to 1} &= \kappa^2 \gamma(-\carrier)\\
\Gamma_{2\to 1} &= 2\kappa^2 \gamma(\carrier)\\
\Gamma_{1\to 2} &= 2\kappa^2 \gamma(-\carrier).
\end{align}

In practice a transmission line is not infinitely long. It can be
rendered effectively infinite by terminating the line by a matched
load~\cite{Pozar05} in the form of a resistor of the same impedance.
Cautiously using that resistor also helps to realize the different
coupling regimes of interest.
On the one hand, by using capactive coupling or impedance mismatch,
the coupling between qubit and line is typically weak, realizing the
weak coupling regime at low temperatures. On the other hand, one can
realize the high temperatures of the singular coupling regime by
mounting the terminating resistor at room temperature as shown in~\cite{Chen11}.

Moreover, note that the flux controls of the GMon maintain the symmetry of the Hamiltonian
and thus of the system-bath coupling up to the point where we get admixtures of the continuum, which can be avoided.
So this essentially means that the calibration of the energy splitting of the
$z$~controls needs to include the Lamb shifts, thus weakly altering the calibration curve.
Finally, in order to avoid limitations on the size of the local $z$-shifts, one may alternatively
introduce individual $x$ and~$y$ microwave controls on every qutrit.


\subsubsection*{Justifying the Born-Markov and the secular approximations}
\hspace{-3.5mm}\begin{minipage}[H!]{0.81\columnwidth}
The physical validity of the Lindblad equation is guaranteed as usual by
separation of the relative time scales in the
system. Let
$\tau_B$ denote the bath correlation time, moreover let
$\tau_S = {2\pi}/\carrier$ be the time scale set by the smallest transition frequency difference
of the rotating frame generator (which in our case is equal to the control carrier frequency) and
take $\tau_R = 1/\gamma_*$ as the time scale of the relaxation, while
$\tau_C = {2\pi}/{|\Omega|}$ shall be the control time scale.
The approximate values used in the GMon setting are given in the table on the right and 
are derived below.

\quad In this setting the Born-Markov approximation holds by $\tau_R \gg \tau_B$, while
the secular approximations hold by virtue of $\tau_R \gg \tau_S$
and $\tau_C \gg \tau_S$.
\end{minipage}
\hspace{5mm}
\begin{minipage}[H!]{0.15\columnwidth}
  \begin{tabular}{ll}
    \hline \hline\\[-4mm]
    time scale & GHz\\
    \hline
    $1/\tau_B$ & $550$\\
    $1/\tau_S$ & $4.8$\\
    $1/\tau_R$ & $\le 0.8$\\
    $1/\tau_C$ & $\le 0.8$\\
    $|J|$ & $0.16$\\
    \hline \hline
  \end{tabular}
\end{minipage}
\vspace{1em}

Let us derive an estimate for the bath correlation time~$\tau_B$.
The ohmic environment with an ultraviolet cutoff at~$\omega_B$ at temperature~$T$
has the following regimes for the decay of the correlation function:
\begin{center}
\begin{tabular}{l l l}
$t\leq1/\omega_B$: & quadratic decay & $1-(\omega_B t )^2$\\[1mm]
$1/\omega_B < t<1/T$: & power-law decay & $(\omega_B t)^{-\zeta}$\\[1mm]
$1/T\leq t$: & exponential decay & $e^{-\omega_B t}$
\end{tabular}
\end{center}

With the UV-cutoff only inducing power-law decay, it may be somewhat surprising to see Markov 
approximation hold at low temperatures. 
The rationale for this is that one invokes the Markov
approximation in the rotating frame, hence demanding $\alpha\omega_B\ll k_BT$ for the decay rate
(with $\alpha$ as a Legget-type Ohmian dissipation).
Thus for given system frequency and temperature, there is always a damping weak enough to make
the environment seem \/`hot\/'.

Taking the more conventional viewpoint that $\tau_B = \omega_B^{-1}$, the question is how to justify
this setting. In most cases, this is due to spurious elements, for example a capacitor shunting out the
tunable coupling inductor at high frequencies. Like any UV-cutoff it is subtle, yet
a safe assumption for $\omega_B$ to be twice the energy gap of nano-scale aluminum,
$\hbar\omega_B = 2\Delta_{Al} \approx 3.5 k_B T_c(Al) \approx k_B \cdot 4.2$~K,
yielding $\omega_B \approx 550$~GHz, and thus
$\tau_B = 1/\omega_B \approx 1.8$~ps.

In our bath model, Eq.~\eqref{eq:bathcoupling}, we use a cutoff of the shape
$(1+x^2)^{-1}$. These are called Drude cutoffs and are typical
characteristics of spurious reactive elements, capacitances or
inductances, shunting the coupling to the environment. While the
precise value of such spurious couplings depend on the fine-tuned
experiment design, a safe (thus somewhat pessimistic) assumption is to set
the spectral density cutoff $\omegacut$
to roughly half the value of $\omega_B$ derived above.
\color{black}


\subsubsection*{Choice of the rotating frame}

\noindent
In the standard derivation of the Lindblad equation in the weak coupling limit, the entire
system Hamiltonian
$H_S = H_0 +\sum_{k=1}^n H_k$
is included in the rotating frame generator~$\Hrf$.
Instead, to simplify the derivation we only include the dominant part~$\sum_{k=1}^n \carrier \: a_k^\dagger a_k$.
This introduces a small perturbation to the dissipator and the Lamb shift.
For the purposes of this study, we justify this choice of~$\Hrf$ as follows:
\begin{enumerate}
\item
To obtain local Lindblad operators, we want to keep $\Hrf$ local,
and thus leave the coupling Hamiltonian~$H_0$ outside~$\Hrf$. This is an acceptable
approximation since $\|H_0\| \ll \|H_S\|$.
\item
We want $\Hrf$ to be time-independent to keep the dissipator and
Lamb shift time-independent. Hence we must also leave the
time-dependent control terms outside~$\Hrf$.
In~\cite{djr2014} the authors present as an example a
case where the time-dependence of the control Hamiltonian results in a
strongly time-dependent dissipator.
We, on the other hand, use resonant control in which the controls
already have to be much weaker than the Hamiltonian they are resonant with
to get rid of the counter-rotating terms.
\item
The level splittings~$\omega_k$ are only
allowed to vary in a small range around~$\carrier$ to justify leaving
the $\omega_k(t)-\carrier$ terms outside~$\Hrf$.
In our control scheme we chose to have only one resonant control
signal (one carrier frequency).
Due to the anharmonicity~$\Delta_k$ the level splittings $\omega_k/(2 \pi)$ must be allowed
to vary $\pm 200$~MHz to enable resonance with
both the $0\leftrightarrow1$ and the $1\leftrightarrow2$ transitions.
Leaving the anharmonicity outside~$\Hrf$
introduces a small error in the Einstein coefficients
involving the state~$\ket{2}$.
\end{enumerate}

\subsubsection*{Numerical results for the GMon setting}

To illustrate the reachability properties of the experimental GMon control system
proposed above we present various numerically optimized results for a system of two GMons, namely
\begin{itemize}
\item[(1)]
initialisation from the maximally mixed state to the ground state,
\item[(2)]
preparation of a maximally entangled two-qutrit GHZ state from the maximally mixed state,
\item[(3)]
preparation of a mixed PPT-entangled two-qutrit state~\cite{Siewert2016,SES16}
from the maximally mixed state, and
\item[(4)]
erasure from the ground state to the maximally mixed state.
\end{itemize}
In these examples we use a Boltzmann factor of $b = 0.001$ which corresponds to a bath
temperature of \mbox{$T_b \approx 35$}~mK,
as well as the idealized zero-temperature case of~$b = 0$.
The sequence durations are given in terms of the inverse geometrized GMon-GMon coupling
$1/J \approx 6.25$~ns, were chosen to be as short as possible while still reaching a suitably low error,
and are all much shorter than the GMon coherence time $T_2\simeq 10 \mu s$.

\begin{table}[h!]
\caption{Summary of Reachability Results for GMons}\vspace{2mm}
\begin{tabular}{l c c c c c c c}
\hline\hline\\[-3mm]
Task & Boltzmann factor $\bboltz = e^{-\beta \hbar \carrier}$ & \quad Duration & \quad Frobenius-norm error $\delta_F$\\[2mm]
\hline\\[-4mm]
(1) Initialisation to ground state & $10^{-3}$ &$10/J$ & $2.63 \times 10^{-3}$\\
\quad -- " -- & 0 & $10/J$ & $1.70 \times 10^{-4}$\\
(2) {\sc ghz}-type state & $10^{-3}$ &$10/J$ & $3.15 \times 10^{-3}$\\
\quad -- " -- & 0 &$10/J$ & $6.99 \times 10^{-4}$\\
(3) {\sc ppt}-entangled mixed state~\cite{Siewert2016,SES16} & $10^{-3}$ &$3/J$ & $< 1 \times 10^{-4}$\\
\quad -- " -- & 0 & $3/J$ & $< 1 \times 10^{-4}$ \\
(4) Erasure to max.~mixed state & $10^{-3}$ &$2/J$ & $< 1 \times 10^{-4}$\\
\hline\hline
\end{tabular}
\end{table}

Note that in the finite-temperature case the
preparation of the pure states (1,2) is necessarily limited in fidelity,
whereas in the zero-temperature case
one can prepare a pure state exactly by exponential decay (i.e.\ in the limit $\gamma \tau \to \infty$
defining the closure of the reachable sets in Theorems~1 and 2).
In contrast, the mixed target states (3,4) that reside in the interior of the set of density operators~$\pos_1$
can be reached exactly in finite time.

\subsubsection*{Sum-up of experimental requirements}
Though the GMon setting described here is readily available and thus lends itself for implementing 
switchable noise, it should not be by no means exceptional, since the only requirements are
\begin{itemize}
\item a unitarily fully controllable system

\item a fast switchable coupling to a dominant noise source, with $1/\tau_R \gg \gamma_0$
	where $\gamma_0$ is the scale of potential unavoidable noise
\item inter-system couplings $J$ that dominate any unavoidable noise, $|J|\gg\gamma_0$,
  to ensure high fidelity, see also Example 1b in App.~\ref{sec:furtherresults}
  (NB: in the GMon setting $J\simeq 160$ MHz, while $1/T_2\simeq 1/10 \mu s = 100$ kHz)
\item a Born-Markov approximation ensured by a bath that equilibrates faster than the
    system relaxes under the switchable noise, $\tau_R \gg \tau_B$.
\end{itemize}
For computational reasons, one usually wants to invoke additional secular approximations
that ignore various rapidly rotating terms in the master equation.

\section{\grape Extended by Incoherent Controls}\label{sec:grape_incoherent}

\noindent
In state transfer problems the fidelity error function used in~\cite{PRA11} is valid
if the purity remains constant, or if the target
state is pure. In contrast to closed systems, in open ones these conditions need not hold. 
Thus here we use a full
Frobenius-norm based error function instead:
$
\delta_F^2 :=
\left\|X_{M:0}-X_\text{target}\right\|_F^2,
$
where $X_{k:0} = X_k \cdots X_1 \text{vec}(\rho_0)$ is the vectorised
state after time slice~$k$,
$X_k = e^{-\Delta t  L_k}$ is the
propagator for time slice~$k$ in the Liouville space, and
$L_k := i\hat H_u(t_k) + \gamma(t_k) \Gamma$.
The gradient of the error is obtained as
\begin{align}
\Partial{\delta_F^2}{u_j(t_k)}
&=
2 \Re \trace\Big((X_{M:0}-X_\text{target})^\dagger \Partial{X_{M:0}}{u_j(t_k)}\Big),\qquad \text{where}\\[3mm]
\Partial{X_{M:0}}{u_j(t_k)}
&= X_M \cdots X_{k+1}\Partial{X_k}{u_j(t_k)} X_{k-1}\cdots X_1 \text{vec}(\rho_0).
\end{align}

\noindent
The exact expression for the partial derivatives of~$X_k$ given in~\cite{PRA11}
requires $L_k$ to be normal, which does not hold in the general case of
open systems of interest here.
Instead we may use, e.g., the finite difference formula to compute
the gradient.
The optimal value of the difference 
is obtained as a trade-off between
the accuracy of the gradient and numerical rounding error, which starts to
deteriorate when the difference becomes very small.
A more preferable option may be to use the auxiliary matrix method~\cite{vLoan78,Kuprov15}
to compute the gradient via series expansions~\cite{DuistermaatKolk2000} based on the formal identity
\be
\exp \left(
\begin{pmatrix}
  -L_k & -i \hat H_j \\
  0 & -L_k
\end{pmatrix} \Delta t
\right)
=
\begin{pmatrix}
  X_k & \Partial{X_k}{u_j(t_k)}\\
  0 & X_k
\end{pmatrix}
\quad
\text{or analogously}
\quad
\exp \left(
\begin{pmatrix}
  -L_k & -\Gamma \\
  0 & -L_k
\end{pmatrix} \Delta t
\right)
=
\begin{pmatrix}
  X_k & \Partial{X_k}{\gamma(t_k)}\\
  0 & X_k
\end{pmatrix}\;.
\ee
An appropriate preconditioning was described in~\cite{DieciPapini2001}.

\section{Further Numerical Results}\label{sec:furtherresults}

\begin{figure}[!ht]
\hspace{7mm}{\sf (a)}\hspace{40mm}\sf{(b)} \hspace{40mm}{\sf (c)}$\hfill$\\
\includegraphics[width=0.24\columnwidth]{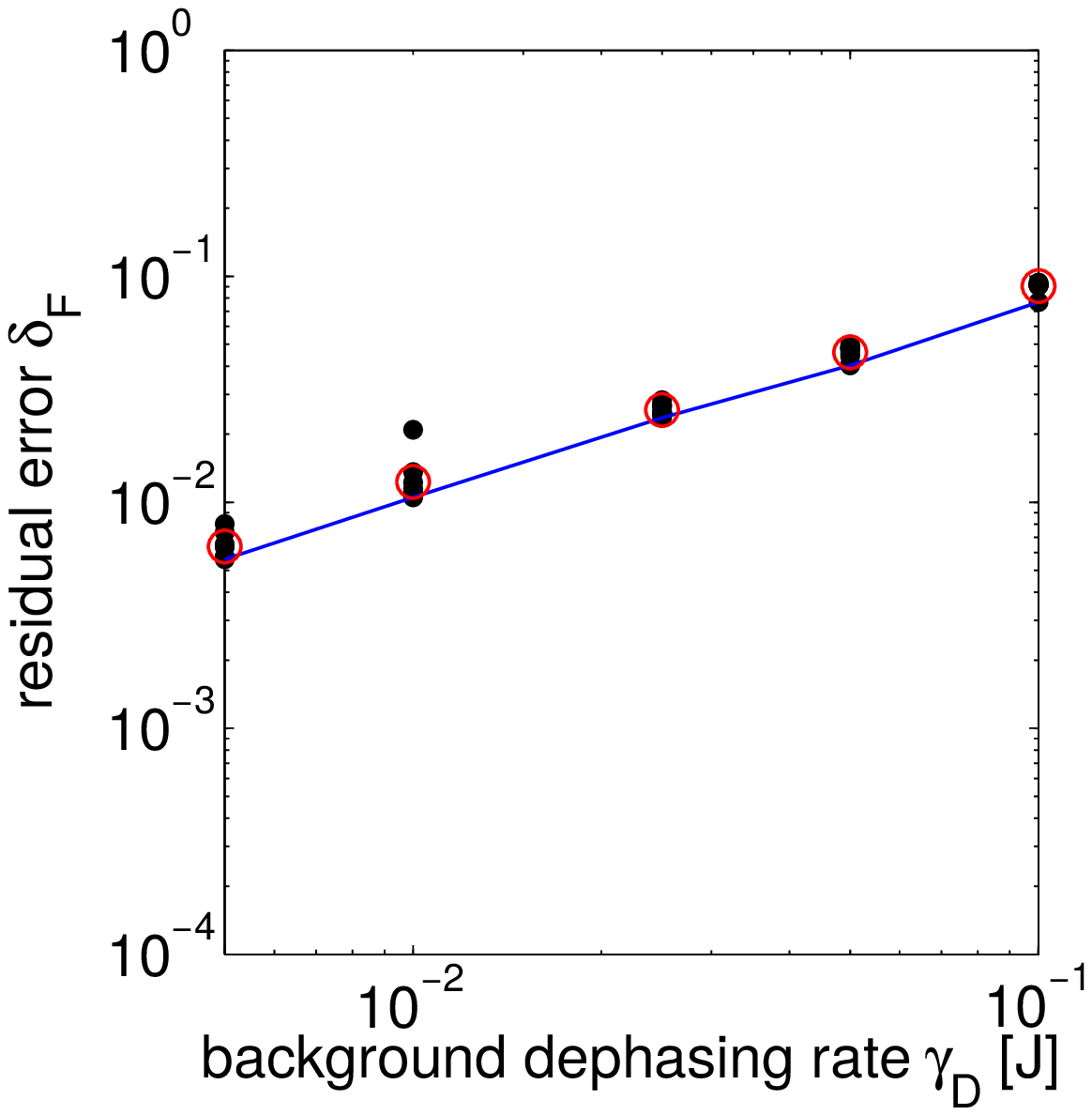}
\includegraphics[width=0.24\columnwidth]{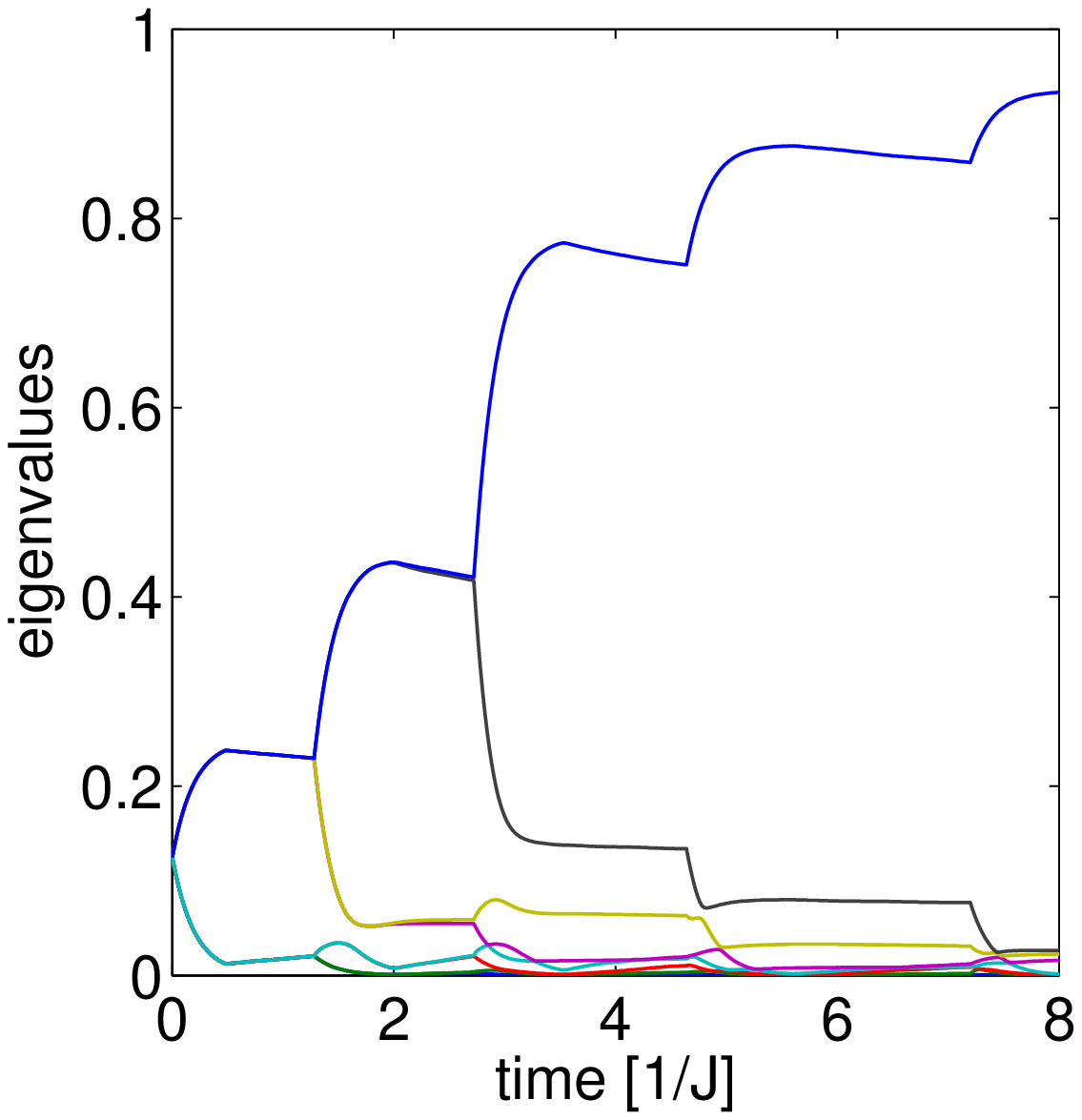}
\includegraphics[width=0.5\columnwidth]{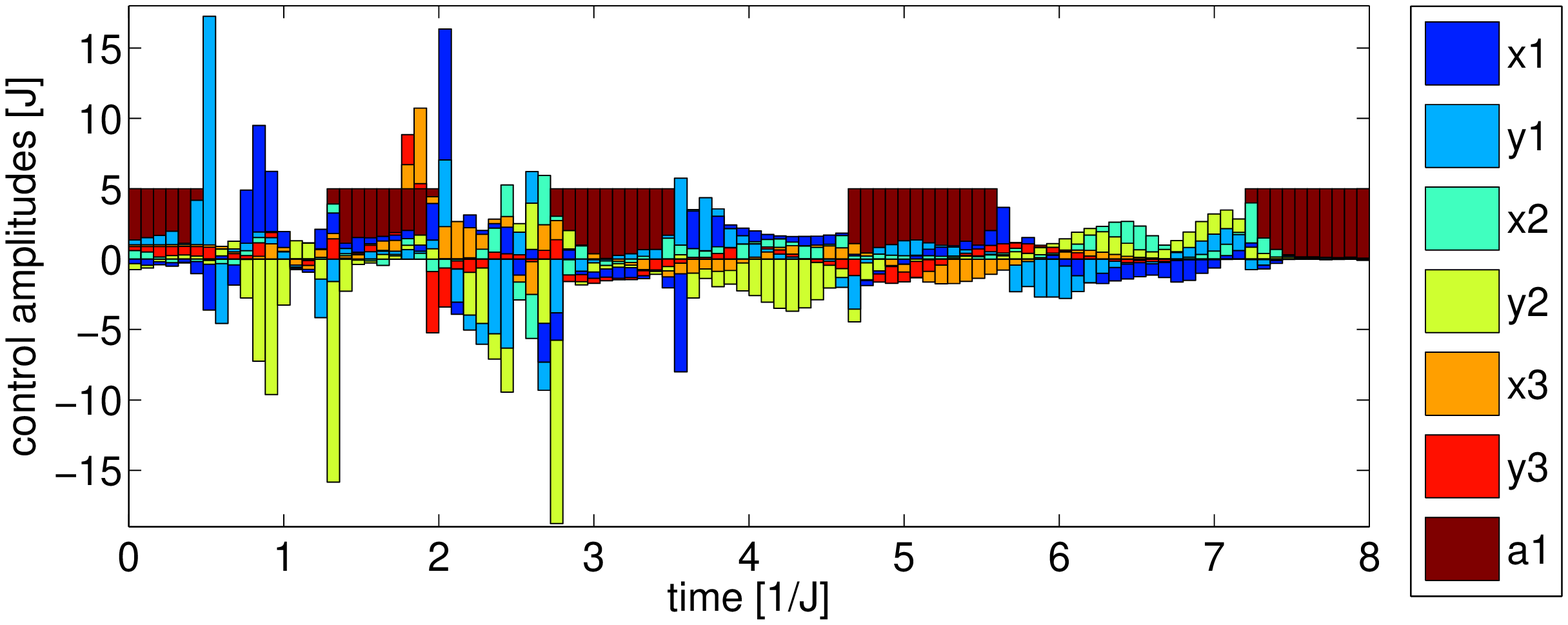}
\caption{\label{fig:th0-dep}
The same initialization-to-$\ket{000}$ task as in Fig.~1 of the main text, 
but with additional non-switchable background
dephasing noise on all the three qubits.
(a)~Error versus dephasing rate~$\gamma_D$ (with $\gamma_D/J\in\{0.005, 0.01, 0.025, 0.05, 0.1\}$)
for sequences of duration~$\tauf=8/J$.
The dots (red circles for averages) are individual numerical optimal control runs with 
random initial sequences.
(b)~Evolution of the eigenvalues under the best sequence for the strongest background noise ($\gamma_D = 0.1 J$)
leading to the zero-state with a considerably low residual error of $\delta_F\approx 0.077$.
This sequence (c) shows five relaxative periods with maximal noise amplitude 
on qubit one ($\gamma_{a1}$) for transforming eigenvalues.
}
\end{figure}

Here we present some further numerical optimization results for
the model Ising control system to complement those in the main paper.

{\bf Example~1b.}
Interestingly, the initialisation task of Example~1 in the main text can still be accomplished
to a good approximation in the presence of unavoidable constant dephasing noise on all the
three qubits.
This is shown in Fig.~\ref{fig:th0-dep} for a range of dephasing rates reaching from
$1\%$ to $20\%$ of the coupling constant. Though the dephasing does not affect the 
evolution of diagonal states, it interferes with the \iSWAP{}s needed to
permute the eigenvalues. For $\gamma_D=0.2~J$, numerical optimal control suggests the
sequence Fig.~\ref{fig:th0-dep} (c) with five dissipative steps and increasing time intervals for the
\iSWAP{}s.

\medskip
{\bf Example~2}.
Here we consider erasing the pure initial
state~$\ket{00\ldots0}$
to the maximally mixed state~$\rho_{\text{th}}$ by controlled bit-flip noise of Eqn.~(3) 
to illustrate the scenario of Thm.~\ref{thm:majorisation}.
For $n$~qubits, one may use a similar $n$-step protocol as in Example~1,
this time approximately erasing each qubit to a state proportional to~$\unity$.
Again one finds that the residual error $\delta_F$  is minimal for equal~$\tau_q$ to give
$
\delta_{F}^2(\expfactorbitflip)
= \tfrac{1}{2^n} \big( \left(1 +\expfactorbitflip^2\right)^n -1\big),
$
where $\expfactorbitflip:= e^{-\gamma_* \tau_q}$.
This yields
\be
\label{eq:0thest}
\tauf_b = \binom{n}{2}\tfrac{1}{J} -\tfrac{n}{2\gamma_*} \ln \big((2^n \delta_{F}^2 +1)^{1/n} -1\big).
\ee
Once again Fig.~\ref{fig:0th}(a) shows that numerical optimal control finds much faster solutions than 
this simplistic protocol.
To better illustrate the qualitative features of the solutions,
we use a weaker noise than in the other examples, with $\gamma_* = 2.5 J$.
Consequently the noise amplitude tends to be maximised throughout the optimised sequence
with the unitaries fully parallelised, as shown in the example
sequence (c), and reflected in the eigenvalue flow (b).
This works so well because
$\rho_{\text{th}}$ is the unique state majorised by every other state,
and thus all admissible eigenvalue transfers lead towards the goal.

The advantage of optimal control -based erasure becomes evident when 
comparing it to free evolution:
Pure bit-flip noise on one qubit (without coherent controls) would just 
average pairs of eigenvalues once if the free evolution
Hamiltonian is a mere Ising-$ZZ$ coupling, which commutes with the initial state.
Hence free evolution does not come closer to the maximally mixed state than  $\delta_F\approx 0.61$
and only by allowing for unitary control, erasure becomes feasible for the Ising chain.
%
%

\begin{figure}[!ht]
\hspace{10mm}{\sf (a)}\hspace{40mm}\sf{(b)} \hspace{40mm}{\sf (c)}$\hfill$\\
\includegraphics[width=0.20\columnwidth]{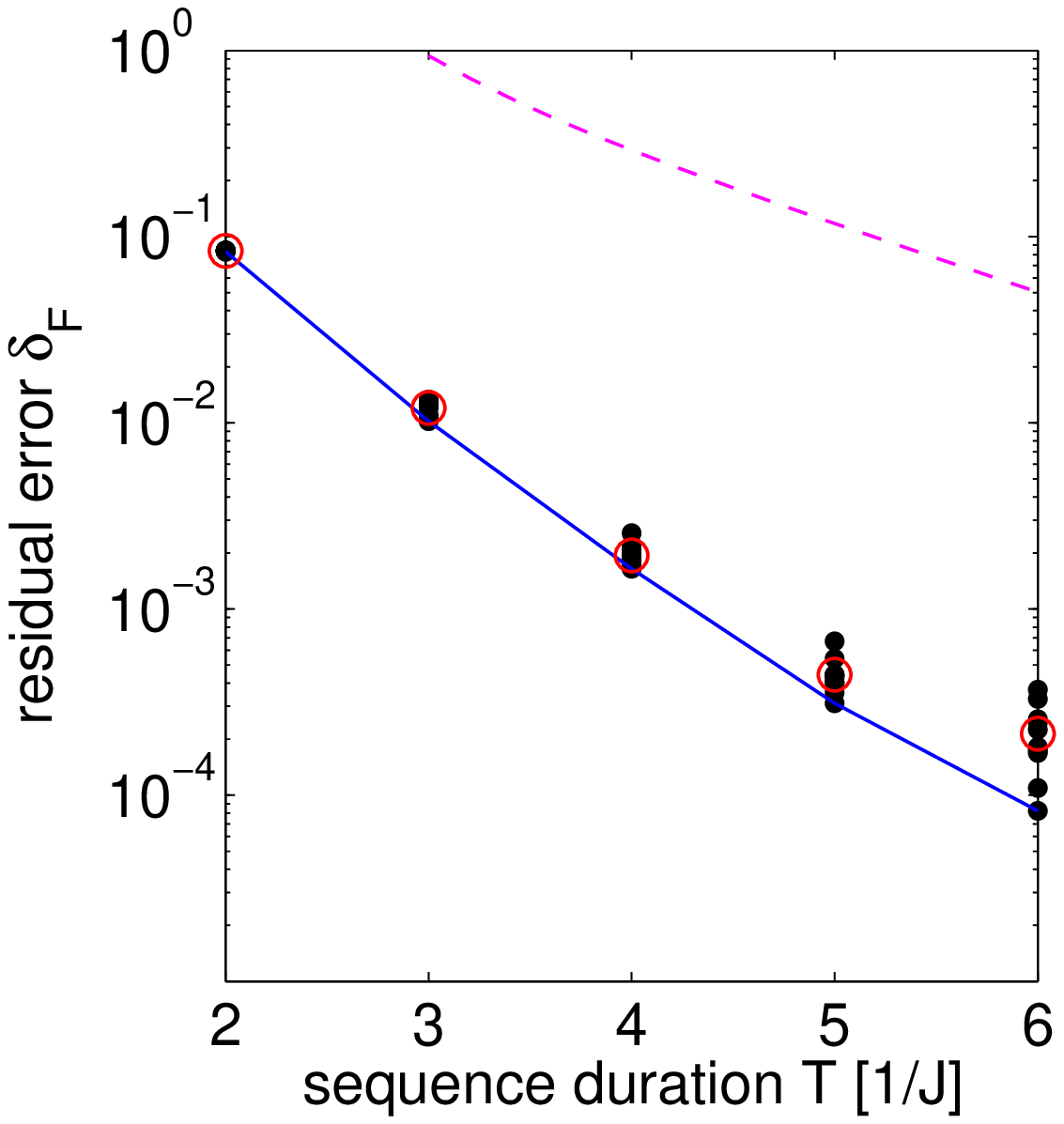}
\includegraphics[width=0.20\columnwidth]{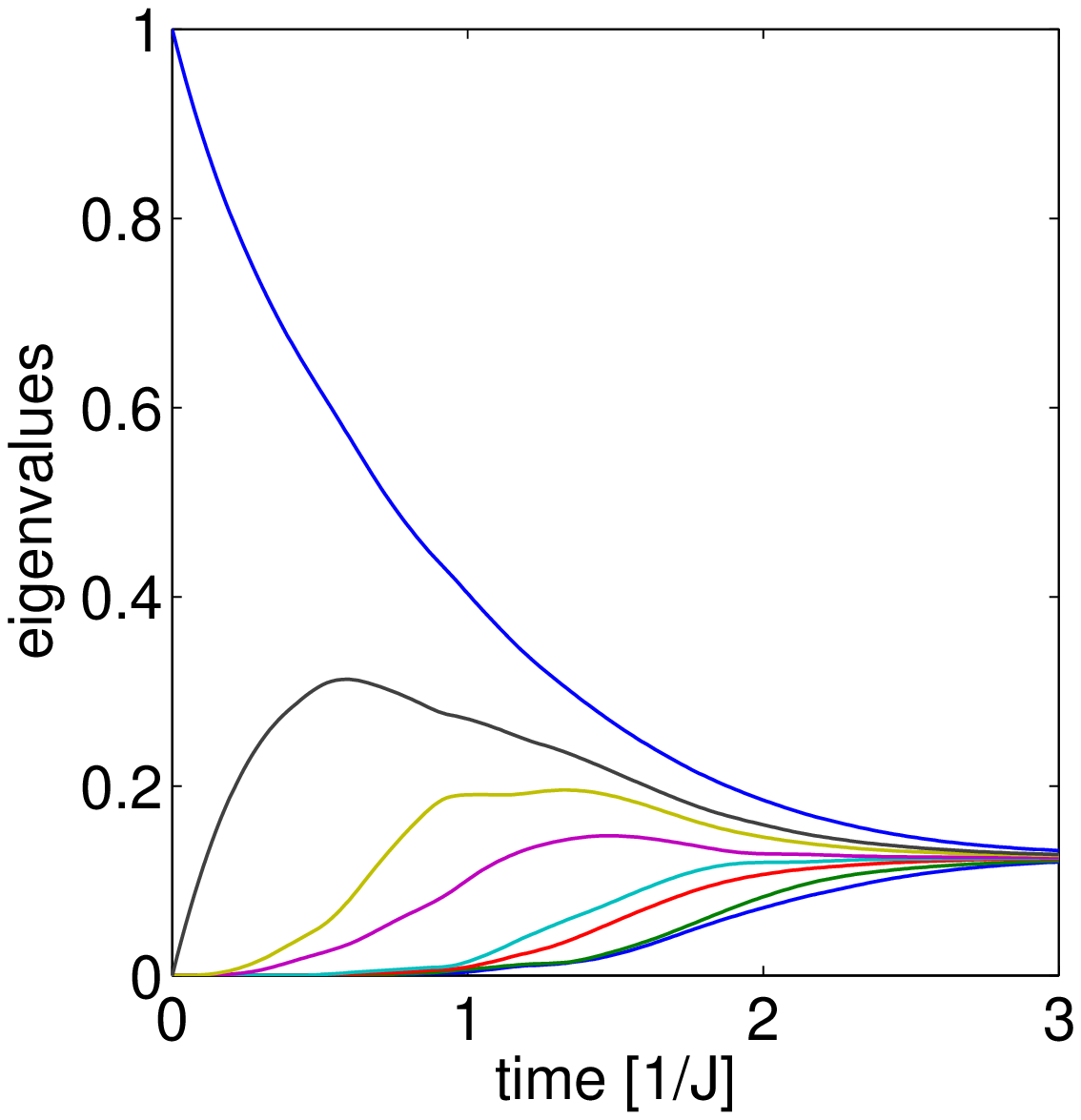}
\includegraphics[width=0.45\columnwidth]{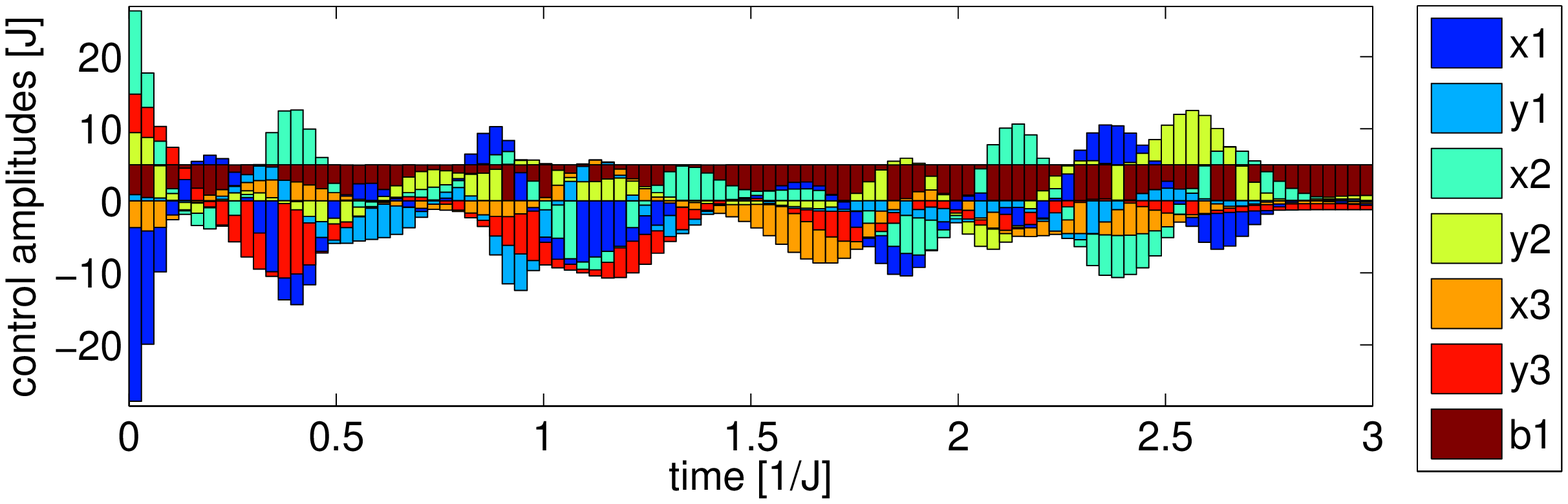}
\caption{\label{fig:0th}
Transfer from the zero-state~$\ket{000}$ to the maximally mixed state~$\rho_\text{th}=\frac{1}{8}\unity$
in a \mbox{3-qubit} Ising-$ZZ$ chain
with controlled bit-flip on qubit one
and local $x,y$-pulse controls on all 
qubits as in {\bf Example~2}.
(a)~Error versus total duration~$\tauf$, with the dashed line as the upper bound from Eqn.~\eqref{eq:0thest}.
Dots (red circles for averages) denote individual numerical optimal control runs with 
random initial sequences.
(b)~Evolution of the eigenvalues under the controls of the best of the
$\tauf=3/J$ solutions.
The corresponding control sequence (c) shows that the noise is always
maximised, and the unitary actions
generated by ($u_{x\nu}, u_{y\nu}$) are fully parallelised with it.
}
\end{figure}

\bigskip
{\bf Example~5}
The final example addresses entanglement generation in a system
similar to the one in~\cite{BZB11}. It consists of
four trapped ion qubits coherently controlled by lasers. 
On top of individual local $z$-controls ($u_{z1},\dots,u_{z4}$) on each qubit, one can
pulse on all the qubits simultaneously by the joint $x$ and $y$-controls
$F_\nu:=\tfrac{1}{2}\sum_{j=1}^4 \sigma_{\nu j}$ with $\nu=x,y$. 
In contrast to an experimental implementation by discrete gates, in our formal 
simplification of the model system, here we also allow 
the quadratic terms $F^2_\nu:=(F_\nu)^2$ to be pulsed continuously together
with the other coherent and incoherent controls. 
All the control amplitudes are expressed
as multiples of an interaction strength $a$.
In contrast to~\cite{BZB11}, where the protocol resorts to an ancilla qubit to be 
added  (following~\cite{VioLloyd01}) for a {\em measurement-based
circuit on the $4+1$ system}, here we do {\em without the ancilla qubit} by
making just the terminal qubit subject to controlled amplitude-damping 
noise with strength $\gamma_{a1}$,
to drive the system from the high-$T$ initial state
$\rho_{\text{th}}:=\tfrac{1}{2^n}\unity$ to the pure entangled target state
$\ket{\text{GHZ}_4} = \tfrac{1}{\sqrt{2}}(\ket{0000}+\ket{1111})$.
As shown in Fig.~\ref{fig:blatt}, the
optimised controls use the noise with maximal amplitude over its entire
duration interrupted just by two short periods of purely unitary control.
%
\begin{figure}[!ht]
\hspace{1.5mm}{\sf (a)}\hspace{95mm}{\sf (b)\hspace{55mm}}$\hfill$\\[0mm]
\includegraphics[width=0.47\columnwidth]{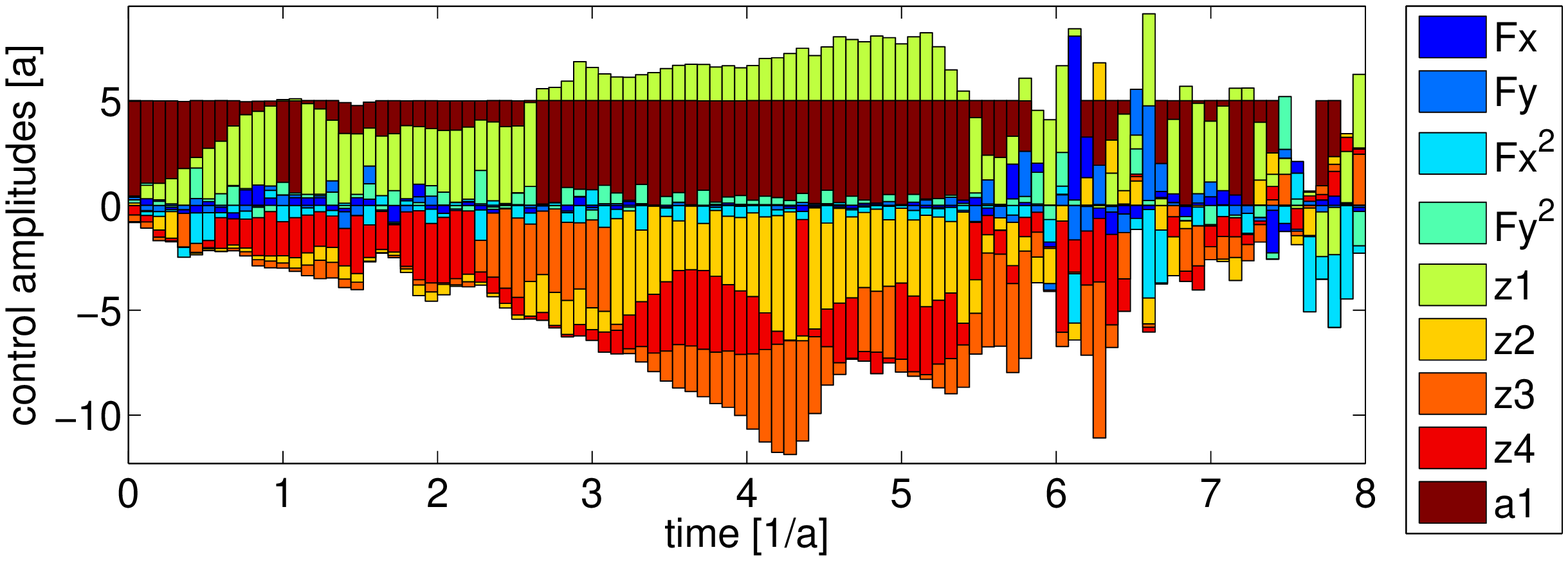}
\hspace{5.5mm}\includegraphics[width=0.47\columnwidth]{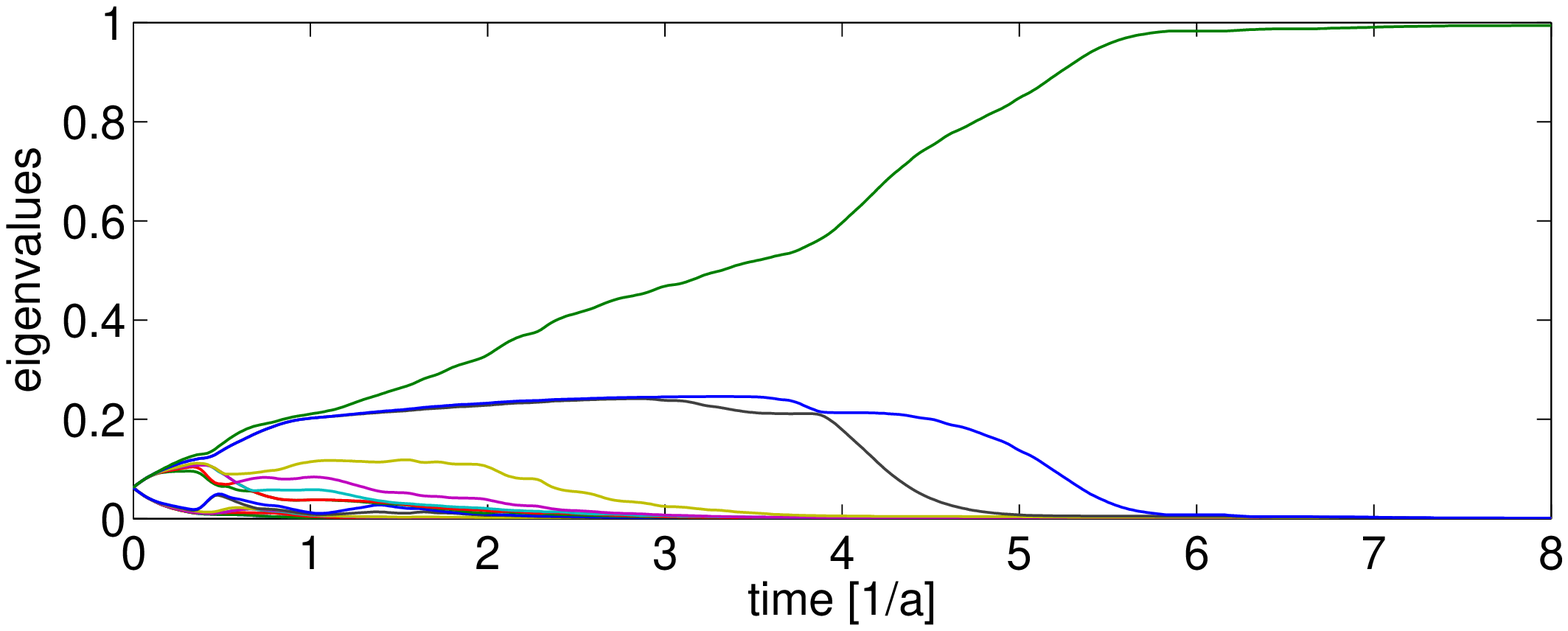}\\[7mm]
\hspace{1.5mm}{\sf (c)}$\hfill$\\[-6mm]
\raisebox{0mm}{\hspace{-3mm}\includegraphics[width=0.28\columnwidth]{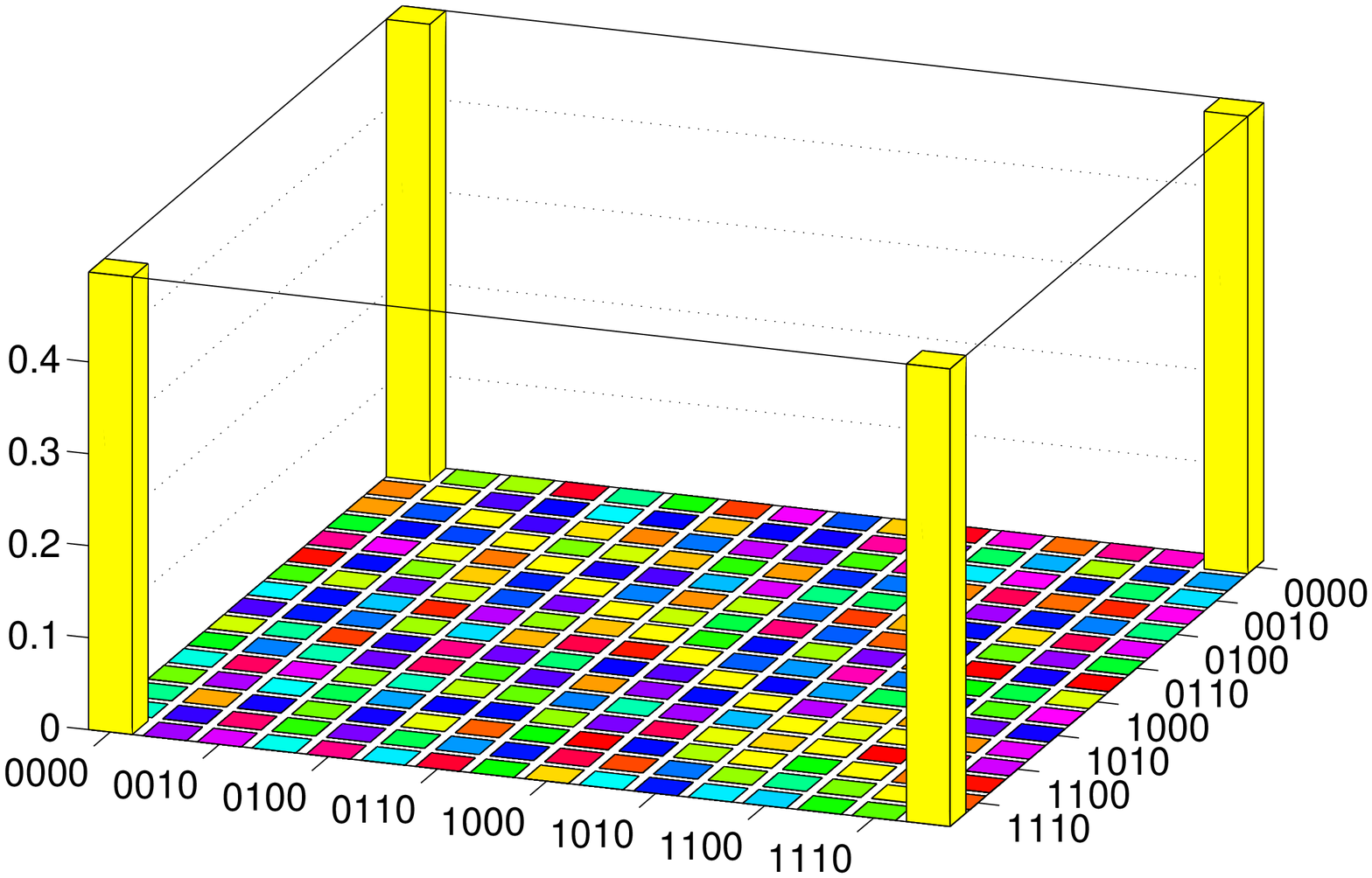}}
\hspace{20mm}
\raisebox{-0mm}{\hspace{3mm}\includegraphics[width=0.28\columnwidth]{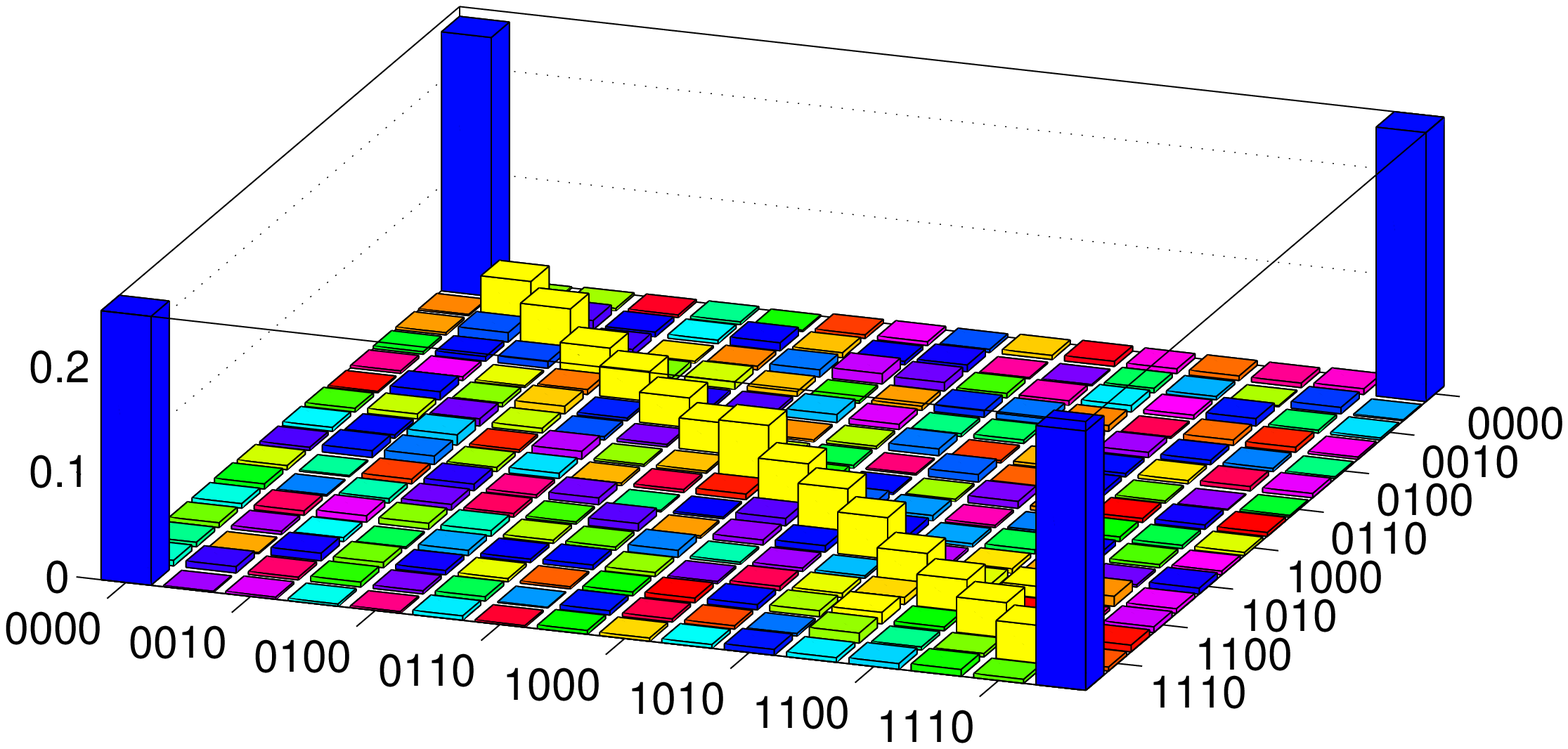}}\\
\raisebox{2mm}{\scriptsize \hspace{52mm} error $\times 100$}\\[0mm]
\caption{\label{fig:blatt}
State transfer from the high-$T$ state to the four-qubit GHZ state in the
ion-trap system of the formal {\bf Example~5} similar to~\cite{BZB11}. By
controlled amplitude damping on one qubit, one
can do without closed-loop measurement-based circuits involving an additional ancilla qubit
as required in~\cite{BZB11,VioLloyd01}.
Our  sequence (a) drives the system
to the state (c), which differs from the target state $\ket{\text{GHZ}_4}$
by an error of $\delta_F\approx 5\times10^{-3}$.
The time evolution of the eigenvalues (b) illustrates parallel action
on all the eigenvalues under the sequence.
}
\end{figure}

\medskip

Finally the difference between optimising amplitude-damping non-unital transfer (as in Thm.~\ref{thm:transitivity})
and bit-flip unital transfer (as in Thm.~\ref{thm:majorisation}) becomes evident:
In the {\em non-unital case}, transitive action on the set of all density operators clearly 
helps to escape from suboptimal intermediate control sequences during the optimisation.
Yet in {\em the unital case}, the majorisation condition 
$\rho_{\rm target}\prec\rho(t)\prec\rho_0$ for all $0\leq t\leq \tauf$
and the boundary conditions $\rho(0)=\rho_0$, $\rho(\tauf)=\rho_{\rm target}$ (at worst for \mbox{$\gamma_* \tauf\to\infty$})
explain potential algorithmic traps: 
one may easily arrive at an intermediate state $\rho_m(t)\prec\rho_0$ that comes closer
to the target state, but will never reach it as it fails to meet
the reachability condition
$\rho_{\rm target}\prec\rho_m(t)$.

\clearpage
\newpage

\section{Scheme for Constructing Majorized Diagonal States Following Hardy, Littlewood, and P{\'o}lya}\label{sec:HLP}

The work of Hardy, Littlewood, and P{\'o}lya~\cite{HLP34} (HLP) can be translated into a {\em constructive}
scheme ensuring the majorisation condition
$\rho_{\rm target}\prec\rho(t)\prec\rho_0\;\text{for all $0\leq t\leq \tauf$}$
to be fullfilled for all intermediate steps. Let
the initial and the target state be given as diagonal vectors with the eigenvalues of the respective density 
operator in descending order, so
$\rho_0=:\diag(y_1, y_2, \dots, y_N)$ and $\rho_{\rm target}=:\diag(x_1, x_2, \dots, x_N)$.
Following~\cite[p32f]{MarshallOlkin}, fix $j$ to be the largest index such that $x_j<y_j$ and let
$k>j$ be the smallest index with $x_k>y_k$. Define $\delta:=\min\{(y_j-x_j),(x_k-y_k)\}$ and
$\lambda:=1-\delta/(y_j-y_k)$. This suffices to construct
\begin{equation}
y':=\lambda y + (1-\lambda) Q_{jk}\; y 
\end{equation}
satisfying $x\prec y' \prec y$. Here the pair-permutation $Q_{jk}$ interchanges the coordinates 
$y_k$ and $y_j$ in $y$.
So  $y'$ is a \mbox{$T$-transform} of $y$, and Ref.~\cite{MarshallOlkin} shows that by $N-1$ successive steps of
$T$-transforming and sorting, $y$ is converted into $x$.
Now the $T$-transforms $\lambda\unity+(1-\lambda)Q_{jk}$
can actually be  brought about by {\em switching on the bit-flip noise} according to 
Eqn.~\eqref{eqn:T-trafo} for a time interval of duration
\begin{equation}\label{eqn:tau-jk}
\tau_{jk}:= - \tfrac{1}{\gamma}\,\ln |\,1-2\lambda\,|\;.
\end{equation}

With these stipulations (and for simplicity assuming a diagonal drift plus Lamb-shift Hamiltonian $H_0+H_{LS}$
to avoid Trotterization as in Corollary~\ref{cor:Trotter-Decoup})
one obtains an iterative analytical scheme for
transferring any $\rho_0$ by unitary control and switchable bit-flip noise on a terminal qubit
into any $\rho_{\rm target}$ satisfying the reachability condition
$\rho_{\rm target}\prec\rho_0$.

\bigskip

\begin{minipage}[H!]{.95\columnwidth}
\begin{center}
\begin{tabular}{c c l}
\hline\hline\\[-3mm]
&& {\bf Scheme for Transferring Any $n$-Qubit Initial State $\rho_0$ into Any
Target State $\rho_{\rm target}\prec\rho_0$ } \\ 
&&{\bf by Unitary Control and Switchable Bit-Flip Noise on Terminal Qubit:}\\[1mm]
\hline\\[-2mm]
&(0) & switch off noise to $\gamma=0$, diagonalise target $U_x\rho_{\rm target}U^\dagger_x=:\diag(x)$
		to obtain vector of eigenvalues \\ 
	&& in descending order $x=(x_1, x_2, \dots, x_N)$; keep $U_x$;\\
&(1) & apply unitary evolution to diagonalise $\rho_0$ and set $\tilde{\rho}_0=:\diag(y)$;\\
&(2) & apply unitary evolution to sort $\diag(y)$ in descending order $y=(y_1, y_2, \dots, y_N)$;\\
&(3) & determine index pair $(j,k)$ by the HLP scheme (described in the text above);\\
&(4) & apply unitary evolution to permute entries $(y_1,y_j)$ and $(y_2,y_k)$ of $y$, so $\diag(y)=\diag(y_j,y_k,\dots)$;\\
&(5) & apply unitary evolution $U_{12}$ of Eqn.~\eqref{eqn:protect} to turn $\rho_y=\diag(y)$ into protected state;\\
&(6) & switch on {\bf bit-flip noise on terminal qubit}  $\gamma(t)=\gamma$
		for duration $\tau_{jk}$ of Eqn.~\eqref{eqn:tau-jk} to obtain $\rho_{y'}$\\
		&&(while decoupling as in Eqn.~\eqref{eqn:Trotter-dec});  \\
&(7) & to undo step (5), apply inverse unitary evolution $U_{12}^\dagger$ to re-diagonalise $\rho_{y'}$ 
		and obtain next iteration of \\ 
	&&  diagonal vector $y=y'$ and $\rho_y=\diag(y)$;\\
&(8) & go to (2) and terminate after $N-1$ loops ($N:=2^n$);\\
&(9) & apply inverse unitary evolution $U_x^\dagger$ from step (0) to take final $\rho_y$ to 
  $U_x^\dagger \rho_y U_x \simeq \rho_{\rm target}$.\\[1mm]
\hline\hline
\end{tabular}
\end{center}
\end{minipage}

\medskip

\noindent
Note that the general HLP scheme need not be time-optimal:
For instance, a model calculation shows that just the dissipative
intervals for transferring $\diag(1,2,3,\dots, 8)/36$ into $\unity_8/8$
under a bit-flip relaxation-rate constant $\gamma=\frac{5}{2} J$ and
achieving the target with $\delta_F=9.95 \times 10^{-5}$ sum up to $\tauf_{\rm relax}=12/J$ in the HLP-scheme,
while a greedy alternative can make it within $\tauf'_{\rm relax}=6.4/J$ and a residual error of $\delta_F=6.04\times 10^{-5}$.

\bigskip

\section{Outlook on the Relation to Extended Notions of Controllability in Open Quantum Systems}\label{sec:controllabilities}
The current results also pave the way to an outlook on controllability aspects of
open quantum systems on a more general scale, since they are much more intricate than in the case of closed systems 
\cite{VioLloyd01,Alt03,Alt04,Rabitz07b,DHKS08,Yuan09,Yuan11,ODS11,Pechen11,KDH12}.

Here we have taken profit from the fact that like in closed systems 
(where pure-state controllability is strictly weaker than full unitary controllability \cite{AA03,SchiSoLea02a}),
in open quantum systems
{\em Markovian state transfer} appears less demanding than the operator lift to the most
general scenario of {\em arbitrary quantum-map generation} (including non-Markovian ones) first connected to
{\em closed-loop feedback} control in \cite{VioLloyd01}. Therefore in view of experimental implementation,
the question arises how far one can get with {\em open-loop} control including noise modulation and whether the border to
{\em closed-loop feedback} control is related to (if not drawn by) Markovianity \footnote{
Note that in favourable cases, one can absorb non-Markovian relaxation by enlarging the system of interest
by tractably few degrees of freedom and treat the remaining dissipation in a Markovian way \cite{PRL_decoh2,JPB_decoh}
}.

\medskip
As used for the mathematical definition of Markovianity in the main text,
due to their defining divisibility properties \cite{Wolf08a,Wolf08b} that allow for an exponential construction
(of the connected component)  as {\em Lie semigroup} \cite{DHKS08}, {\em Markovian} quantum maps 
are indeed a well-defined special case of the more general completely positive trace-preserving 
(CPTP) semigroup of Kraus maps, which clearly comprise non-Markovian ones, too. 
While some controllability properties of general Kraus-{\em map generation}
have been studied in \cite{VioLloyd01,Rabitz07b}, 
a full account of controllability notions in open systems should also encompass {\em state-transfer} 
to give the following major scenarios:
\begin{enumerate}
\item Markovian state-transfer controllability ({\MSC}),
\item Markovian map controllability (\MMC),
\item general (Kraus-map mediated) state controllability (\KSC) (including the infinite-time limit of \/`dynamic state controllability\/' (\DSC) \cite{Rabitz07b,Pechen11}),
\item general Kraus-map controllability (\KMC) \cite{VioLloyd01,Rabitz07b}. 
\end{enumerate}
Writing \/`$\subseteq$\/' and \/`$\subsetneqq$\/' in some abuse of language for \/`weaker than\/' and
\/`strictly weaker than\/', one obviously has at least
$
\MSC \subseteq \KSC\; \text{and}\; \MMC \subseteq \KMC, 
$
while $\DSC\subsetneqq\KMC$ was already noted in the context of control directly over the Kraus operators \cite{Rabitz07b}.
In pursuing control over environmental degrees of freedom, 
Pechen \cite{Pechen11, Pechen12} also proposed a scheme, where both coherent plus incoherent light 
(the latter with an extensive series of spectral densities depending on ratios over the difference of eigenvalues 
of the density operators to be transferred) 
were shown to suffice for interconverting arbitrary states with non-degenerate eigenvalues in their density-operator
representations.

Yet the situation outlined above is more subtle, since unital and non-unital cases may differ. 
In this work, we have embarked on unital and non-unital Markovian state controllability, 
\MSC\ \footnote{For simplicity, first we only consider
	the extreme case of non-unital maps (such as amplitude damping) allowing for pure-state fixed points and 
           postpone the generalised cases
	parameterised by $\theta$ in Appendix~\ref{app:B} till the very end.}: 

Somewhat surprisingly, in the {\em non-unital} case (equivalent to amplitude damping),
the utterly mild conditions of unitary controllability plus
bang-bang switchable noise amplitude on one single internal qubit (no ancilla) suffice for acting transitively on 
the set of all density operators (Theorem~\ref{thm:transitivity}). Hence these features fulfill the maximal condition $\KSC$ already. 
In other words, for cases of non-unital noise equivalent to amplitude damping (henceforth indexed by \/`{\sf nu}\/'),
$\KSC_{nu}$ implies $\KSC$.
Moreover, under the reasonable assumption that the mild conditions in Theorem~\ref{thm:transitivity} 
are in fact the {\em weakest}
for controlling Markovian state transfer $\MSC_{nu}$ in our context, Theorem~\ref{thm:transitivity} shows that
$\MSC_{nu}$ implies $\KSC$ via $\KSC_{nu}$. 
So in the (extreme) non-unital cases, there is no difference between Markovian and 
non-Markovian state controllability. ---
On the other hand in order to compare non-unital with unital processes, taking  Theorems~\ref{thm:transitivity} and \ref{thm:majorisation} together
proves $\MSC_{u}\subsetneqq\MSC_{nu}$, since the former is restricted by the majorisation condition (see
Theorem~\ref{thm:majorisation}).

Similarly, in the {\em unital} case (equivalent to bit-flip), the mild conditions of unitary controllability plus
bang-bang switchable noise amplitude on one single internal qubit suffice for achieving {\em all}  state transfers obeying majorisation
(Theorem~\ref{thm:majorisation}). Hence again they fulfill the maximal condition $\KSC_u$ at the same time. 
This is because  state transfer under {\em every} unital CPTP Kraus map (be it Markovian or non-Markovian)  has to meet 
the majorisation condition; so we get $\KSC_u$. On the other hand, the majorisation condition itself imposes the
restriction $\KSC_u\subsetneqq\KSC$.
Again, under the reasonable assumption that the mild conditions in Theorem~\ref{thm:majorisation} 
are in fact the weakest
for controlling Markovian state transfer $\MSC_{u}$ in our context, Theorem~\ref{thm:majorisation} shows that
$\MSC_{u}$ implies $\KSC_u$. Thus also in the unital case, there is no difference between Markovian and 
non-Markovian state controllability.

The results on these two cases, i.e.\  the non-unital and the unital one (in the light of Appendix~\ref{app:B} seen as the limits $\theta=0$ and
$\theta=\tfrac{\pi}{2}$, respectively), can therefore be summarized as follows:
\begin{corollary}
In the two scenarios of Theorem~\ref{thm:transitivity} (non-unital) and \ref{thm:majorisation} (unital), 
Markovian state controllability already implies Kraus-map mediated state controllability and one finds
\begin{equation}
\begin{array}{c c c c c}
\MSC_{nu} & \Longrightarrow & \KSC_{nu}  & \Longrightarrow & \KSC \\[2mm]
\bigcup\negthickspace\nparallel & &\bigcup\negthickspace\nparallel & 
	\begin{turn}{-45}\raisebox{1.8mm}{$\bigcup\negthickspace\nparallel$}\end{turn} &\\[2mm]
\MSC_u & \Longrightarrow & \KSC_u & &\\ 
\end{array}
\end{equation}
\end{corollary} 
 However, whether $\MSC_\theta\Longrightarrow\KSC_\theta$ also holds in the
finite-temperature generalisation of Appendix~\ref{app:B}, where
$\theta$ can range over the entire interval $\theta\in[0,\tfrac{\pi}{2}]$ (with $\theta=0$ giving the limiting cases
$\MSC_{nu}, \KSC_{nu}$ and $\theta=\tfrac{\pi}{2}$ yielding $\MSC_u, \KSC_u$), currently remains an open question.

This has an important consequence for experimental implementation of {\em state transfer} in open quantum systems: 
On a general scale in $n$-qubit systems, unitary control plus measurement-based closed-loop feedback from one 
resettable ancilla (as, e.g.,  in Ref.~\cite{BZB11} following \cite{VioLloyd01})
can be replaced by unitary control plus open-loop bang-bang switchable non-unital noise (equivalent to amplitude damping)
on a single internal qubit. This is because both scenarios are sufficient to ensure Markovian and non-Markovian
state controllability \KSC. {\bf Example~5} in the main part illustrates this general simplifying feature.

\medskip

The privileged situation of quantum-state transfer and simulation can be elucidated by comparison to
the afore-mentioned more demanding task of {\em quantum-map synthesis}: Even in the connected component
of quantum maps (i.e.\ arbitrarily close to the identity) there exist non-Markovian maps which thus cannot
be constructed exponentially. More precisely, Wolf and Cirac identified a class of indivisible single-qubit
channels \cite{Wolf08a} (ibid.\ Thm.~23 on rank-three channels with diagonal Lorentz normal forms),
which were shown to extend into the connected component \cite{DHKS08}.
Since (at least) those maps cannot be constructed exponentially and thus do not follow a 
Lindblad master equation, they serve as easy counter examples excluding that \MMC already implies \KMC.
So Markovian map controllability is strictly weaker than Kraus-map controllability, i.e.\ $\MMC\subsetneqq \KMC$.

Yet some questions with regard to the operator lift to map synthesis remain open:
Assessing a demarcation
between \MMC and \KMC (and their unital versus non-unital variants) seems to require 
different proof techniques than used here. 
In a follow-up study we will therefore further develop our lines of assessing the differential geometry 
of Lie semigroups in terms and their Lie wedges \cite{DHKS08,ODS11} to this end,
since judging upon Markovianity on the level of Kraus maps is known to be more intricate \cite{Wolf08b,Wolf12}.
More precisely,
{\em time-dependent Markovian} channels come with a general form of a Lie wedge
in contrast to {\em time-independent Markovian} channels, whose generators form the special structure of a
Lie semialgebra (i.e.\ a Lie wedge closed under Baker-Campbell-Hausdorff multiplication). In  \cite{DHKS08},
we have therefore drawn a detailed connection between these differential properties of Lie semigroups and the 
different notions of divisibility studied as a defining property of Markovianity in the seminal work \cite{Wolf08a}.

Again, these distinctions will decide on simplest experimental implementations in the sense
that measurement-based {\em closed-loop feedback} control may be required for non-Markovian maps in \KMC, 
while {\em open-loop} noise-extended control may suffice for Markovian maps in \MMC.
More precisely, 
{\em closed-loop feedback} control was already shown to be {\em sufficient} for \KMC  in \cite{VioLloyd01}
(which was the aim that work set out for), yet it remains to be
seen {\em to which extent} it is also {\em necessary}.
Wherever it turns out to be unnecessary, measurement-based closed-loop feedback control on a system extended by one 
resettable ancilla \cite{VioLloyd01,BZB11,SMB12} would be not be stronger than our open-loop scenario of full unitary control extended by (non-unital) noise modulation. This has direct bearing on the simplification of quantum simulation 
experiments \cite{SMB12}. Therefore a demarcation line between Markovian and non-Markovian maps in differential 
geometric terms will be highly useful.



\end{document}